\documentclass{article}[11pt]
\usepackage{graphicx}
\usepackage{amsmath}
\usepackage{amsthm}
\usepackage{amssymb}
\usepackage[letterpaper, margin=1in]{geometry}
\usepackage{algorithm}
\usepackage{algpseudocode}
\usepackage{caption} 
\usepackage{thmtools,thm-restate}
\usepackage{cite}
\theoremstyle{plain}

\newtheorem{lemma}{Lemma}
\newtheorem{cor}{Corollary}
\newtheorem{defn}{Definition}
\newtheorem{prop}{Proposition}
\newtheorem{remark}{Remark}

\usepackage{rotating}
\usepackage{tikz}
\usepackage{pgf,pgffor,pgfplots,pgfplotstable}
\usepgfmodule{shapes,plot}
\usetikzlibrary{positioning,decorations,arrows,shapes.misc,calc}

\tikzset{cross/.style={cross out, draw=black, minimum size=2*(#1-\pgflinewidth), inner sep=0pt, outer sep=0pt}, cross/.default={1pt}}
\newcommand*{\BigLat}{2}%
\newcommand*{\NSmCell}{4}%
\newcommand*{\NBigCellhoriz}{1}%
\newcommand*{\NBigCellvert}{4}%
\pgfmathsetmacro{\SmLat}{\BigLat / \NSmCell}%
\pgfmathsetmacro{\NSmCellhoriz}{\NBigCellhoriz*\NSmCell}%
\pgfmathsetmacro{\NSmCellminus}{\NSmCell-1}%
\pgfmathsetmacro{\NSmCellvert}{\NBigCellvert*\NSmCell}%
\pgfmathsetmacro{\NSmCellsq}{\NSmCell*\NSmCell}%

\tikzset{cross/.style={cross out, draw=black, minimum size=2*(#1-\pgflinewidth), inner sep=0pt, outer sep=0pt}, cross/.default={1pt}}
\renewcommand*{\L}{10}%
\newcommand*{\La}{2.5}%
\newcommand*{\Lb}{5}%
\newcommand*{\N}{40}%
\newcommand*{\Ma}{10}%
\pgfmathsetmacro{\DELTA}{\L/\N}%
\newcommand*{\LocList}{1, 23, 35 }
\newcommand*{\LocListAliasa}{1,  3, 5 } 
\newcommand*{\LocListAliasb}{1, 3, 15 } 
\newcommand*{\LocListSearchb}{1, 3, 5, 11, 13, 15 }
\newcommand*{\LocListSearchc}{1, 3, 15, 21, 23, 35}

\newcommand*{\LocListAliasShuffled}{13,15,19}

\newcommand*{\LocListSearchShuffled}{33,35,39}

\newcommand*{\LocListReduced}{0,1,2,3,4,5,6,7,8,9}

\author{
  Letourneau, Pierre-David\thanks{Corresponding author:
    letourneau@reservoir.com} \and Langston, M. Harper \and  Meister, Benoit \and Lethin, Richard \thanks{Reservoir Labs, 632 Broadway Suite 803, New York, NY 10012}}
\title{A sparse multidimensional FFT for real positive
  vectors\thanks{This research was developed with funding from the
    Defense Advanced Research Projects Agency (DARPA). The views,
    opinions and/or findings expressed are those of the author and
    should not be interpreted as representing the official views or
    policies of the Department of Defense or the U.S. Government.
    Patent Pending. (U.S. Patent Application No. 62/286,732)} } 
\begin{document}
\maketitle
\begin{abstract}
We present a sparse multidimensional FFT (sMFFT) randomized algorithm
for real positive vectors. The algorithm works in any fixed dimension, 
requires ($\mathcal{O}\left( R \log(R)
\log\left( N \right )  \right )$) samples and runs in
$\mathcal{O}\left( R \log^2(R) \log\left( N \right )  \right
)$ complexity (where $N$ is the total size of the vector in $d$ dimensions and $R$ is the number
of nonzeros). It is stable to low-level noise and exhibits an exponentially
small probability of failure. 
\end{abstract}
\newpage
\section{Introduction}
The Fast Fourier Transform (FFT) algorithm  reduces the computational
cost of computing the Discrete Fourier Transform (DFT) of a general
complex $N$-vector from $\mathcal{O}(N^2)$ to $\mathcal{O}(N \log(N)
)$. Since its popularization in the 1960s \cite{cooley1965algorithm},
the FFT algorithm has played a crucial role in
multiple areas including scientific computing \cite{darden1993particle},
signal processing \cite{mallat1999wavelet} and computer science
\cite{frigo1998fftw}. In the general case, such scaIing is at most a factor $\log(N)$ from optimality. In more restricted cases however, such as when the vector to 
be recovered is sparse, it is possible to significantly improve on the
latter.

Indeed, the past decade or so has seen the design and study of various
algorithms that can compute the DFT of sparse vectors using significantly less
time and  measurements than traditionally required
\cite{pawar2013computing,mansour1995randomized,lawlor2013adaptive,kushilevitz1993learning,iwen2010combinatorial,goldreich1989hard,gilbert2002near,akavia2003proving,boufounos2012s,akavia2010deterministic,plonka2016sparse,plonka2015deterministic,gilbert2005improved,indyk2014sample,indyk2014sparse,indyk2014nearly,ghazi2013sample,hassanieh2012nearly,hassanieh2012simple,akavia2014deterministic,iwen2013improved}. That
is, if $f$ is an $N\times 1$ vector corresponding to the DFT of an
$N\times 1$ vector $\hat{f}$ containing at most $R \ll N$ nonzero elements,
it is possible to recover $\hat{f}$ using significantly fewer samples than the
traditional ``Nyquist rate'' ($\ll \mathcal{O}(N)$) 
and in computational complexity much lower than that of the FFT ($\ll
\mathcal{O}(N\log(N))$). These schemes are generally referred to as ``sparse Fast Fourier
Transform'' (sFFT) algorithms, and they generally fall within two categories: 1) deterministic\cite{akavia2010deterministic,akavia2014deterministic,iwen2010combinatorial,iwen2013improved,lawlor2013adaptive}
versus randomized algorithms, and 2) exactly versus approximately sparse \cite{mansour1995randomized,gilbert2002near,
gilbert2005improved,indyk2014sample,indyk2014sparse,indyk2014nearly,ghazi2013sample,hassanieh2012nearly,hassanieh2012simple,akavia2014deterministic,iwen2013improved};  for a periodic signal,
\begin{equation}
\label{1DsFFT:orig_func}
f(x)  =  \sum_{j=0}^{N-1} e^{-2 \pi i \, x \, j}\, \left ( \hat{f}_j + \eta \hat{\nu}_j \right )  \, ,
\end{equation}
we say that it is \emph{$R$-sparse} if the support of the spectrum $\hat{f}$, i.e., the set of indices for which $\hat{f}_j \not = 0$, has cardinality smaller than or equal to $R$. Recovery of the location and magnitude of the nonzero coefficients is referred to as the \emph{exactly} $R$-sparse FFT problem if $\eta = 0$ and as the \emph{noisy} $R$-sparse FFT problem if $0 < \eta$ is relatively small. When the signal is not $R$-sparse but nonetheless compressible \cite{candes2006stable}, the problem of finding the best $R$-sparse approximation is referred to as the \emph{approximately} sparse FFT problem and takes the following form: given an accuracy parameter $\epsilon$, find $\hat{f}^*$ such that, 
\begin{equation}
\label{l2l2}
|| \hat{f}^* - \hat{f} ||_a \leq (1+\epsilon)\,  \min_{ \hat{y} \, : \,  \mathrm{R-sparse} } || \hat{y} - \hat{f} ||_b + \eta || \hat{f} ||_c
\end{equation}
where $|| \cdot ||_a$, $|| \cdot ||_b$ and $|| \cdot ||_c$ are $\ell_p$-norms (generally $p=1$ or $2$).

Of the former category, randomized algorithms have had the most 
success in practice thus far; although some deterministic algorithms do exhibit quasilinear complexity in $R$ and polylogarithmic complexity in $N$, the algorithmic constant and the exponents are often so large that the methods are only competitive when $N$ is impractically large or when the sparsity satisfies stringent conditions. On the other hand, both the algorithmic constant and the complexity of randomized sparse FFTs are in general much smaller. The most recent results on the topic are shown in Table \ref{tab:alg_prop}.

 Among approximate sparse FFT algorithms, the best complexities achieved so far have been obtained by \cite{indyk2014nearly} and \cite{ghazi2013sample} in the average and worst-case scenario respectively. Both are randomized algorithms but use different techniques to locate the support of the important frequencies. \cite{ghazi2013sample} uses a technique known as orthogonal frequency division multiplexing (OFDM) where the problem is ``lifted" from 1D to 2D and sparsity along each column/row is leveraged, whereas \cite{indyk2014nearly} uses binning techniques (shifting, shuffling and filtering) to produce a support-locating estimator. Both methods ultimately rely on estimating some complex phase containing information with regard to the support of the signal.

Among exact sparse FFT algorithms, the best complexities achieved so far where obtained by \cite{iwen2010combinatorial} in the worst case, and  \cite{lawlor2013adaptive} (deterministic) and \cite{ghazi2013sample} (randomized w/ constant probability of success) both in the average case. We note however that \cite{iwen2010combinatorial} can be very unstable, even to low-level noise, being based upon the Chinese  Remainder Theorem. As for \cite{lawlor2013adaptive}, its performance degrade significantly when a worst case scenario is considered, becoming of the same order as \cite{iwen2010combinatorial}. In addition, to achieve such low complexity, \cite{lawlor2013adaptive} makes use of adaptive sampling, by which samples for a subsequent iterations are chosen based on information available at the current iteration in some. By contrast, our algorithm uses non-adaptive sampling. These results show an interesting trade-off: in an $R$-sparse worst-case scenario, algorithms exhibiting the best sampling/computational complexity often possess scaling much inferior than their counter-part per regards to the probability of failure $p$, and vice-versa. Such characteristics can be detrimental in settings where high-performance is needed and failure cannot be tolerated.
 
Finally, we also observe from Table \ref{tab:alg_prop} that whenever the algorithm is generalizable to multiple dimensions, the algorithmic constant exhibits a scaling which is at least exponential in the dimension $d$.

In this light, this work presents an algorithm for treating the noisy $R$-sparse FFT problem when the spectrum $\hat{f}$ is nonnegative\footnote{Complex signals are treated in a sister paper: \emph{``A sparse multidimensional fast Fourier transform for complex vectors.''}. The complex case exhibits the same scaling as the real positive case except that $C_{\mathrm{complex}} (p)  \sim  \log^2 (p^{-1} )$. See Appendix \ref{complexvec} for a high-level description of the approach.}. The proposed algorithm possesses a low sampling complexity of $\tilde{\mathcal{O}}\left( R \log(R)  \log \left( N \right )  \right )$ and a computational complexity of $\tilde{\mathcal{O}}\left( R \log^2(R) \log \left( N \right )  \right )$ (where $\tilde{O}$ indicates the presence of $\log^c(\log(\cdot))$ factors), and further scales like $\log(p)$ with respect to the probability of failure $p$, thus alleviating the issues underlined above. In addition, our method possesses constant scaling with respect to dimension, i.e., in $d$ dimensions with $N=M^d$ unknowns, the scaling is of the form $\log(N) = d \log(M)$ without further dependence. The algorithm is also non-adaptive.

The proposed method uses tools similar to those introduced in \cite{indyk2014nearly}, especially its shuffling and filtering techniques. However, our support-locating scheme does not rely on sample-based phase estimation; it works in Fourier space rather than in sample space. Furthermore, rather than proceed simultaneously in locating and computing the value the coefficients of the``heavy" frequencies, we completely separate both processes. Intuitively, locating the support without attempting to estimate values accurately is a simpler task that can be carried out more rapidly and with higher probability of success than attempting to perform both tasks at once. As for the treatment of the high-dimensional problem, we demonstrate how it is in fact possible to reduce any $d$-dimensional problem to a $1$-dimensional problem without any overhead cost through the use of rank-$1$ lattices. 

\begin{center} 
{ 
\begin{table}[H]
\begin{center}
\begin{tabular}{|c|c|c|c|c|c|c|} \hline 
Reference &  Time & Samples & $C(p)$ & $C(d)$ & Type & Model  \\ \hline
\cite{indyk2014nearly} & $\tilde{\mathcal{O}} \left ( R \log^2 (N)\right ) $ & $\tilde{\mathcal{O}}\left ( R \log(N) \right )  $ & $ p^{-2}$  & $1D$  & approx. & worst    \\ \hline
\cite{indyk2014sample} & $\mathcal{O}\left ( N \log^3 (N)  \right ) $ & $ \mathcal{O} \left ( R \log(N) \right ) $  & $ \sim N^{-\mathcal{O}\left (R \right )}$   & $d^d$ & approx. & worst  \\ \hline
\cite{ghazi2013sample} & $\mathcal{O}\left ( R \log^2 (N)  \right ) $ & $ \mathcal{O} \left ( R \log(N) \right ) $  & expected & $1D$, $2D$ & approx.  & average \\  \hline

\cite{gilbert2002near} & $\mathcal{O}(R^2 \log^{O(1)}(N) ) $ & $\mathcal{O}(R  \log^{O(1)}(N) ) $ & $ \log(p^{-1}) $  & $ 2^{\mathcal{O}(d)}$ & approx.   & worst    \\ \hline
\cite{gilbert2005improved} & $\mathcal{O}(R \log^{O(1)}(N)) $ & $\mathcal{O}(R \log^{O(1)}(N)) $ & $ \log(p^{-1}) $  & $ 2^{\mathcal{O}(d)}$ & approx.  & worst    \\ \hline
\cite{iwen2010combinatorial} & $ \tilde{\mathcal{O}} \left ( R \log (R) \log^4(N)\right ) $ & $ \tilde{\mathcal{O}} \left ( R \log^4(N)\right )  $ & $\log\left( p^{-1} \right )$  & $1D$ & approx. & worst   \\  \hline
\cite{kapralov2016sparse} & $\mathcal{O}(R \log^{d+3}(N) )$ & $\tilde{\mathcal{O}}( R  \log(N) ) $   &  $ \sim \frac{1}{\log(N)} $  & $2^{\mathcal{O}(d^2)}$ & exact  & worst  \\ \hline
\cite{hassanieh2012nearly} & $\mathcal{O} \left ( R \log (N)\right ) $ & $\mathcal{O}\left ( R \log(N) \right )  $ & constant  & $1D$ & exact  & average   \\  \hline
\cite{iwen2010combinatorial} & $ \tilde{\mathcal{O}} \left ( R^2 \log^4 (N) \right ) $ & $ \tilde{\mathcal{O}} \left ( R^2 \frac{\log^4(N)}{\log(R)} \right )  $ & deterministic  & $1D$ & exact & worst   \\  \hline
\cite{ong2015fast} & $\mathcal{O}( R \log^4 (N) ) $ & $ \mathcal{O}(R \log^3 (N) )$ & $\sim \frac{1}{R \log^3 (N)}  $  & $2D$ & noisy & worst  \\ \hline
\cite{lawlor2013adaptive} & $\mathcal{O} \left ( R \log (R)\right ) $ & $\mathcal{O}\left ( R \right )  $ & deterministic  & $1D$ &  noisy & average \\  \hline
\cite{ghazi2013sample} & $\mathcal{O}\left ( R \log (R)  \right ) $ & $ \mathcal{O} \left ( R \right ) $  & constant & $1D$, $2D$ & noisy  & average \\ \hline \hline
\hline
This paper & $\tilde{\mathcal{O}}\left( R \log^2(R) \log \left( N \right )  \right )$ & $\tilde{\mathcal{O}}\left( R \log(R)  \log \left( N \right )  \right )$  &  $\log(p^{-1} ) $   & $ \mathcal{O}(1)$ & noisy & worst \\ \hline
\end{tabular}
\end{center}
\caption{  \footnotesize Computational characteristics of recent sparse FFT algorithms. $C(d)$: behavior of algorithmic constant with respect to dimension $d$. $C(p)$: behavior of algorithmic constant with respect to probability of failure $p$; $\sim$ indicates the dependence of the probability of failure on other parameters. ``exact" implies the exactly $R$-sparse FFT problem without any noise whereas ``noisy" implies the presence of low-level noise only.}
\label{tab:alg_prop}
\end{table}%
}\end{center}

The paper is structured as follows: in Section \ref{overview}, we
introduce the notation and a description of the
problem. In Section \ref{1DsFFT}, we describe the algorithm in the noiseless
one-dimensional case. The case of noisy data is discussed in Section
\ref{stability}. We describe how to convert between one dimension and multiple dimensions in Section \ref{MD}. Finally, numerical results are provided in Section \ref{Numer}. All proofs can be found in
Appendix~\ref{appendix:proofs}, and a discussion of the generalization to complex vectors can be found in Appendix \ref{complexvec}.   

\section{Statement of the problem and preliminaries}
\label{overview}
In this section, we introduce the notation used throughout the
remainder of the paper. Unless otherwise stated, we consider a 1D function $f(x)$ of the form, 
\begin{equation}
\label{1DsFFT:orig_func}
f(x)  =  \sum_{j=0}^{N-1} e^{-2 \pi i \, x \, j}\, \left ( \hat{f}_j + \eta \hat{\nu}_j \right )  \, ,
\end{equation}
for some finite $0 < N \in \mathbb{N}$ and noise level $0 \leq \sqrt{N}  || \hat{\nu} ||_2 \leq \eta$. It is further assumed  that the vector $\hat{f}$ has real nonnegative elements, i.e., $\hat{f}_j \geq 0, \, \forall j $, and that its support,
\begin{equation*}
\mathcal{S} := \left \{ j \in \{0,1, ..., N-1 \} : |\hat{f}_j| \not = 0   \right \}
\end{equation*}
satisfies $0 \leq  \# \mathcal{S} \leq R < N < \infty$, where $\#$ indicates cardinality. In particular, we are
interested in the case where $R \ll N$. Given some accuracy parameter $\epsilon$ above the noise level $\eta$, the problem involves computing an $R$-sparse vector $f^*$ such that,
\begin{equation*}
|| f^* - \hat{f} ||_2 \leq \epsilon \, || \hat{f}||_2
\end{equation*}

We shall denote by
$\mathcal{F}$ the Fourier transform (and $\mathcal{F}^*$ its
inverse/adjoint), i.e., 
$$\mathcal{F}\left [ \hat{f} (\xi) \right ] (x)
= \int_{\mathbb{R}^d} e^{-2\pi i \, x\cdot\xi } \, \hat{f} (\xi) \,
\mathrm{d} \xi$$
where $d$ represents the ambient dimension. The
size-$N$ Discrete Fourier Transform (DFT) is defined as
\begin{equation}
\label{notation:dftdef}
f_{n;N} =\left [ F_N \, \hat{f} \right ]_n  =  \sum_{j=0}^{N-1} e^{-2 \pi i  \frac{n \, j}{N}  } \, \hat{f}_j  , \; n = 0,1, ..., N-1.
\end{equation}

\section{A sparse FFT in 1D}
\label{1DsFFT}
In this section, we describe a fast way to compute the one-dimensional
DFT of a bandlimited and periodic function $f(x)$ of the form of
Eq.\eqref{1DsFFT:orig_func}. Our approach to
this problem can be broken into two separate steps: in the first step,
the support $\mathcal{S}$ of the vector $\hat{f}$ is recovered, and in
the second step, the nonzero values of $\hat{f}$ are computed using
the knowledge of the recovered support. We describe the algorithm in
the noiseless case 
in this section, followed by a discussion of its stability to noise
in Section \ref{stability}. Pseudo-code is provided in Algorithms
\ref{1DsFFT:algorithm}-\ref{1DsFFT:compute_values}. 

\begin{algorithm}[H]
\caption{1DSFFT($R,N,p$)}
\label{1DsFFT:algorithm}
\begin{algorithmic}[1]
\State Let $\mu$, $\Delta$ and $\eta$ be estimates for $\min_{j \in \mathcal{S} } |\hat{f}_j| $,$\frac{|| \hat{f}||_\infty}{\mu}$ and the noise $\sqrt{N} \, || \hat{\nu} ||_2$ respectively.
\State (In the noiseless case, let $\eta$ be the desired level of accuracy)
\\ $\mathcal{S} \leftarrow \mathrm{FIND\_\,SUPPORT}(R,N,p,\mu,\Delta, \eta)$
\\ $\hat{f} \leftarrow \mathrm{COMPUTE\_\,VALUES}(\mathcal{S},R,N,p,\mu,\Delta,\eta)$
\\ Output: $\hat{f}, \mathcal{S}$.
\end{algorithmic}
\end{algorithm}

\subsection{Finding the support} 
For the remainder of this section, refer
to the example in Figure \ref{1DsFFT:findsupport}. From a 
high-level perspective, our support-finding scheme uses three major
ingredients: 1)sub-sampling, 2)shuffling and 3)low-pass
filtering. Sub-sampling reduces the size of the problem to
a manageable level, but leads to aliasing. Nonetheless, when the nonnegativity assumption is satisfied, an aliased Fourier
coefficient is nonzero \emph{if and only if} its corresponding aliased
lattice contains an element of the true support  (note that positivity
is crucial here to avoid cancellation). This provides a useful
criterion to discriminate between elements that belong to the support
and elements that do not.  

To help the reader better understand the scheme, we proceed through an example and refer to  Figure~\ref{1DsFFT:findsupport}. To begin with, consider $k,N,M_k \in 
\mathbb{N}$, $0< \alpha <1$ and $ \mathcal{S}_k,\mathcal{W}_k,
\mathcal{M}_k \subset \{ 0,1,..., N-1\}$. We define the following,
\begin{itemize}
\setlength\itemsep{0em}
\item the \emph{aliased support} $\mathcal{S}_k$ at step $k$
  corresponds to the indices of the elements of the true support $\mathcal{S}$ modulo
  $M_k$;
\item the \emph{working support} at step $k$ corresponds to the set
  $\mathcal{W}_k := \{ 0, 1, ..., M_k - 1\}$;
\item a \emph{candidate support} $\mathcal{M}_k$ at step $k$ is any set satisfying $ \mathcal{S}_k \subset \mathcal{M}_k \subset \mathcal{W}_k$ of size
  $\mathcal{O}(\rho R \log(R))$. 
\end{itemize}

Line {\bf 0)} (Figure \ref{1DsFFT:findsupport}) represents a lattice
(thin tickmarks) of size,
$$N=40 = 5\,  \prod_{i=1}^3 2 = K \,\prod_{i=1}^P \rho_i$$
which contains only $3$ positive frequencies (black dots; $\mathcal{S} = \{1, 23,35 \}$).  In the beginning, (step
$k=0$) only the fact that $\mathcal{S} \subset \{0, 1, ..., N-1\}$ is known. The first step ($k=1$) is performed as follows: letting,
$$M_1 = \frac{N}{\prod_{i=2}^P \rho_i }= \rho_1 K =\mathcal{O}(R \log(R) )$$
sample the function $f(x)$ at,
$$x_{n_1
  \prod_{i=2}^P \rho_i ; \, N} = \frac{n_1 \,\prod_{i=2}^P \rho_i}{N}
= \frac{n_1}{M_1}  =  x_{n_1; \, M_1}$$
to obtain,  
\begin{align}
\label{1DsFFT:aliasing}
f_{n_1 \prod_{i=2}^P \rho_i  ;\, N} &=    \sum_{j=0}^{N-1} e^{-2 \pi i
  \, \frac{n_1 \prod_{i=2}^P \rho_i  \, j}{N} }\, \hat{f}_j =
\sum_{l=0}^{M_1-1} e^{-2 \pi i \, \frac{n_1 \,  l}{M_1} }\, \left (
\sum_{j :  j \mathrm{mod} M_1 = l } \hat{f} _j\right ) =
\sum_{l=0}^{M_1-1} e^{-2 \pi i \, \frac{n_1 \,  l}{M_1} } \, \hat{f}
_j^{(1)} =  f_{n_1 ; \, M_1}  
\end{align}
for $n_1 \in \mathcal{M}_1 := \{ 0,1,..., M_1-1 \} $ defined as the
candidate support in the first step. The samples
correspond to a DFT of size $M_1$ of the vector 
$\hat{f}^{(1)} $ with entries that are an \emph{aliased version} of
those of the original vector $ \hat{f}$. 
These can be computed through the FFT in order $\mathcal{O}(M_1
\log(M_1)) =  \mathcal{O}(R \log^2 (R))$.  In this first
step, it is further possible to rapidly identify the aliased support
$\mathcal{S}_1$ from the knowledge of $\hat{f}^{(1)}$ since the former
correspond to the set,
 $$ \{ l \in \{ 0, 1, ..., M_1 - 1\} : \hat{f}^{(1)}_l \not = 0  \}$$
due to the fact that $$\hat{f}^{(1)}_l :=  \sum_{j :  j \mathrm{mod} M_1 = l } \hat{f} _j >0 \Leftrightarrow l \in \mathcal{S}_1$$ following the nonnegativity assumption. In our example, $M_1 = \rho_1 K = 2 \cdot 5= 10$ which leads to
\begin{align*}
\mathcal{S}_1 &= \{1\, \mathrm{mod} \, 10 , \,  23\, \mathrm{mod} \,
10, \, 35\, \mathrm{mod} \, 10 \} = \{ 1,3,5 \} =  \{ l \in \{ 0, 1, ..., 9\} : \hat{f}^{(1)}_l \not = 0  \} \\
 \mathcal{W}_1 &= \mathcal{M}_1 = \{ 0,1,..., 9\}.
\end{align*}
This is shown on line {\bf 1)} of Figure~\ref{1DsFFT:findsupport}. For
this first step, the working support $\mathcal{W}_1$ is equal to the
candidate support $\mathcal{M}_1$.

\begin{figure}[H]
\begin{tikzpicture}

\draw[line width = 1.pt] (0,0) -- (\L,0);
\foreach \x in {0,...,\N}{
	\draw[line width = 0.5pt] (\DELTA*\x,-0.1) -- (\DELTA*\x, 0.1);
	}
	
\renewcommand*{\k}{1};
\foreach \x in \LocList{
	\draw[fill=black] (\DELTA*\x, 0) circle  (2pt);
	\k = \k+1;
}
\draw (0,-0.1) node[below] {$0$};
\draw (\L,-0.1) node[below] {$N =40$};
\draw (-1,0.) node {\bf 0)};

\end{tikzpicture}

\vspace*{5pt}

\begin{tikzpicture}

\draw[line width = 1.pt] (0,0) -- (\La,0);
\foreach \x in {0,...,\Ma}{
	\draw[line width = 0.5pt] (\DELTA*\x,-0.1) -- (\DELTA*\x, 0.1);
	}
	
\renewcommand*{\k}{1};
\foreach \x in \LocListAliasa{
	\draw[line width = 2.pt] (\DELTA*\x,-0.1) -- (\DELTA*\x, 0.1);
	\k = \k+1;
}
\draw (0,-0.1) node[below] {$0$};
\draw (\La,-0.1) node[below] {$ M_1 =10$};
\draw (-1,0.) node {\bf 1)};

\end{tikzpicture}

\begin{tikzpicture}

\draw[line width = 1.pt] (0,0) -- (\Lb,0);
\foreach \x in \LocListSearchb{
	\draw[line width = 0.5pt] (\DELTA*\x,-0.1) -- (\DELTA*\x, 0.1);
	}
	
\renewcommand*{\k}{1};
\foreach \x in \LocListAliasb{
	\draw[line width = 2.pt] (\DELTA*\x,-0.1) -- (\DELTA*\x, 0.1);
	\k = \k+1;
}
\draw (0,-0.1) node[below] {$0$};
\draw (\Lb,-0.1) node[below] {$M_2 = 20$};
\draw (-1,0.) node {\bf 2)};

\end{tikzpicture}


\begin{tikzpicture}

\draw[line width = 1.pt] (0,0) -- (\L,0);
\foreach \x in \LocListSearchc{
	\draw[line width = 0.5pt] (\DELTA*\x,-0.1) -- (\DELTA*\x, 0.1);
	}
	
\renewcommand*{\k}{1};
\foreach \x in \LocList{
	\draw[line width = 2.pt] (\DELTA*\x,-0.1) -- (\DELTA*\x, 0.1);
	\k = \k+1;
}
\draw (0,-0.1) node[below] {$0$};
\draw (\L,-0.1) node[below] {$N =40 $};
\draw (-1,0.) node {\bf 3)};

\end{tikzpicture}

\caption{\footnotesize Computing the support $\mathcal{S}$. Line {\bf 0)}:
  Initialization; (unknown) elements of $\mathcal{S}$ correspond to black dots
  and lie in the grid $\{ 0, 1, ..., N-1\}$. Line {\bf 1)}: First
  step; elements of the candidate support $\mathcal{M}_1$ are
  represented by thin tickmarks and those of the aliased support
  $\mathcal{S}_1$ by thick tickmarks. $\mathcal{S}_1$ is a subset of
  $\mathcal{M}_1$ and both lie in the working support $\{ 0, 1, ...,
  M_1-1\}$. Line {\bf 2)}: Second step; elements of the candidate
  support $\mathcal{M}_2$ correspond to thin tickmarks and are
  obtained through de-aliasing of $\mathcal{S}_1$. Elements of the
  aliased support $\mathcal{S}_2$ correspond to thick tickmarks. Both
  lie in the working support $\{ 0, 1, ..., M_2-1\}$. $M_2$ is a
  constant factor of $M_1$.  Line {\bf 3)}: The final step correspond to the step when the working is equal to $\{ 0,1 , ..., N-1\}$.} 
\label{1DsFFT:findsupport}
\end{figure}
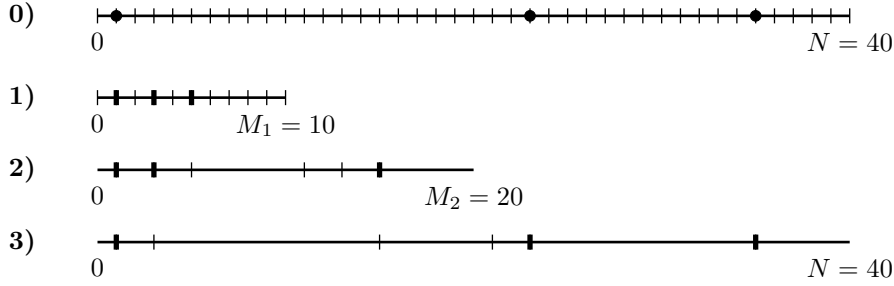

Then, proceed to the next step ($k=2$) as follows: let,
$$M_2 = \rho_2 M_1 = K \, \prod_{i=1}^2 \rho_i = 5 \cdot 2^2 = 20$$
and consider the samples,
\begin{align*}
f_{n_2 \prod_{i=3}^P \rho_i ; \, N}  = \sum_{l=0}^{M_2-1} e^{-2 \pi i \, \frac{n_2 l }{M_2} }\, \left ( \sum_{j :  j \mathrm{mod} M_2 = l } \hat{f}_j\right ) = \sum_{l=0}^{M_2-1} e^{-2 \pi i \, \frac{n_2 \,  l}{M_2} }\,  \hat{f}^{(2)}_l =f_{n_2; \, M_2}  
\end{align*}
for $n_2 = 0,1,..., M_2-1$ as before. Here however, knowledge of
$\mathcal{S}_1$ is 
incorporated. Indeed, since $ M_2$ is a multiple of $M_1$, it follows
upon close examination that,
$$\mathcal{S}_2 \subset \cup_{k=0}^{\rho_1
  - 1} \left (  \mathcal{S}_1 +   k M_1 \right ) :=
\mathcal{M}_2 .$$
That is, the set $\mathcal{M}_2$, defined as the union
of $\rho_1 = \mathcal{O}(1)$ translated copies of $\mathcal{S}_1$,
must itself contain $\mathcal{S}_2$. Furthermore, it is of size
$\mathcal{O}(\rho_1 \, \# \mathcal{S}_1  ) = \mathcal{O}(\rho R \log(R))$ by
construction. It is thus a proper candidate support (by
definition). In our example, one obtains 
\begin{equation*}
\cup_{k=0}^1 \left (  \mathcal{S}_1 +   k M_1 \right ) = \left \{
  1,3,5 \right \} \cup \left \{ 1+10, 3+10, 5+10  \right \} = \left \{
  1,3,5,11,13,15  \right \} = \mathcal{M}_2 ,
\end{equation*}
which contains the aliased support,
$$\mathcal{S}_2 = \{1\,
\mathrm{mod} \, 20 , \,  23\, \mathrm{mod} \, 20, \, 35\, \mathrm{mod}
\, 20 \} = \{ 1,3,15 \}$$
as shown on line \textbf{2)} of Figure
\ref{1DsFFT:findsupport}. The working support becomes
$\mathcal{W}_2 := \{0, 1, ..., 19 \} $. Once again, it is possible to
recover $\mathcal{S}_2$ by leveraging the fact that $\{ l \in \{ 0, 1,
..., M_2 - 1\} : \hat{f}^{(2)}_l  \not = 0  \} = \mathcal{S}_2$. Here
however, the cost is higher since computing $\hat{f}^{(2)}$ involves
performing an FFT of size $M_2 = 20$.  Continuing in the fashion of the first step, the
cost would increase  
exponentially with $k$, so additional steps are required to contain the
cost. Such steps involve a special kind of shuffling and filtering of the samples followed by an FFT, and we describe this in
detail in Section \ref{subsec:rapid_recovery} below. Altogether, it is shown that $\mathcal{S}_k$ can now be recovered from the knowledge
of $\mathcal{M}_k$ at any step $k$ using merely $\tilde{\mathcal{O}}(R \log(R) )$ samples and
$\tilde{\mathcal{O}}(R \log^2 (R))$ computational steps. 

Following the rapid recovery of $\mathcal{S}_2$, we proceed in a
similar fashion until $\mathcal{W}_k := \{ 0,1, ...,N-1\} $
at which point $\mathcal{S}_k = \mathcal{S}$. Throughout this process, the size of the aliased support $\mathcal{S}_k$ and candidate
support $\mathcal{M}_k$ remain of order $\mathcal{O}(R \log(R))$ while the size of
the working support increases exponentially fast; i.e.,
$$\#\mathcal{W}_k = \mathcal{O}(K \, \prod_{i=1}^k \rho_i  ) \geq 2^k
\cdot R.$$
This therefore implies $P = \mathcal{O} \left ( \log\left ( \frac{N}{R} \right )  \right )$ ``dealiasing'' steps, and thus a total cost of $\tilde{\mathcal{O}} \left ( R \log(R) \log \left ( \frac{N}{R}  \right ) \right )$ samples and $\tilde{\mathcal{O}} \left ( R \log^2(R) \log \left ( \frac{N}{R}  \right ) \right )$ computational steps to identify $\mathcal{S}$. The steps of this support-recovery algorithm are described in Algorithm~\ref{1DsFFT:findsupport_alg}, the correctness of which is
guaranteed by the following proposition,  
\begin{prop}
In the noiseless case, Algorithm \ref{1DsFFT:findsupport_alg} outputs
$\mathcal{S}$, the support of the nonnegative $R$-sparse vector $\hat{f}$ with probability at least $(1-p)$ using $\tilde{\mathcal{O}} \left ( R \log(R) \log \left ( \frac{N}{R}  \right ) \right )$ samples and $\tilde{\mathcal{O}} \left ( R \log^2(R) \log \left ( \frac{N}{R}  \right ) \right )$ computational steps. 
\end{prop}
\begin{proof}
Refer to Algorithm \ref{1DsFFT:findsupport_alg} and Proposition
\ref{1DsFFT:findsupport:findsupport_proof}
(Appendix~\ref{appendix:support}) as well as the above discussion. 
\end{proof}
From the knowledge of $\mathcal{S}$, it is possible to recover the
actual values of $\hat{f} $ rapidly and with few samples. This is the
second major step of the sMFFT which we describe below in Section \ref{1DsFFT:CV}.

\begin{algorithm}[H]
\caption{ $\mathrm{FIND\_\,SUPPORT}(R, \tilde{N},p,\mu, \Delta, \eta)$}
\label{1DsFFT:findsupport_alg}
\begin{algorithmic}[1]
\State Pick $2 \leq \rho $ and $0 < \delta \ll 1$ such that $\eta \leq \frac{\delta \, \mu}{2}$ .
\State Let $\alpha = \frac{1}{\rho}$, $K = \frac{\max \{ 8 , \frac{2}{\alpha}  \} }{\pi} \, R \, \sqrt{\log\left ( \frac{2R \Delta}{\delta} \right ) \log \left(  \frac{2 \Delta}{\delta} \right )} $ and choose $N = K \, \prod_{i=1}^P \rho_i \geq \tilde{N}$ where $2 \leq \rho_i  \leq \rho \; \forall \; i$
\State Let $M_1 = K$ and $ \mathcal{M}_1 := \left \{ 0,1,..., M_1-1  \right \} $.
\For{ $k$ from $1$ to $P$ }
	\State $\mathcal{S}_k \leftarrow \mathrm{FIND\_ALIASED\_SUPPORT}(\mathcal{M}_{k}, M_k, K,\alpha,p, \delta, \mu,  \Delta)$
	\State $ \mathcal{M}_{k+1} := \cup_{m=0}^{\rho_k-1} \left (  \mathcal{S}_{k} +   m M_{k} \right )$.
	\State $M_{k+1} := \rho_k \,M_k$
	\State $\alpha \leftarrow \frac{1}{\rho_i}$
\EndFor
\State Output: $\mathcal{S}_P$.
\end{algorithmic}
\end{algorithm}

\subsubsection{Rapid recovery of $\mathcal{S}_k$ from knowledge of
  $\mathcal{M}_k$.}
\label{subsec:rapid_recovery}
Details are given here as to how to solve the
problem of rapidly recovering the aliased support $\mathcal{S}_k$ from
the knowledge of a candidate support $\mathcal{M}_k$. Before
proceeding, a few definitions are introduced.  
\begin{defn}
\label{1DsFFT:Fourierset}
Let $ 1 \leq K \leq M \in \mathbb{N}$. Then, define the set $\mathcal{A}(K;M)$ as,
\begin{equation*}
\mathcal{A}(K;M) := \left \{ m  \in \{ 0,1,...,M-1 \} \, : \,m \leq \frac{K}{2} \; \mathrm{or} \;  \left |  m - M \right | <  \frac{K}{2} \right \} 
\end{equation*}
\end{defn}
\begin{defn}
\label{1DsFFT:Qset}
Let $0< M \in \mathbb{N}$. Then, we define the set $\mathcal{Q}(M)$ as,
$$\mathcal{Q}(M):= \left  \{ q
\in [0,M)  \cap \mathbb{Z} : q 
  \perp M \right \}, $$
  where the symbol $\perp$ between two integers
  indicates they are coprime. 
\end{defn}
Algorithm \ref{RASR:recoveraliasedsupp} shows how to solve the aliased
support recovery problem rapidly; correctness is
guaranteed by Proposition \ref{RASR:correctness}, which relies on
Proposition \ref{1DsFFT:findsupport:findsupport_proof}
(Appendix~\ref{appendix:support}). Proposition
\ref{1DsFFT:findsupport:findsupport_proof} states that if the elements of
an aliased vector of size $M_k$ with aliased support $\mathcal{S}_k$
containing at most $R$ nonzeros are shuffled (according to appropriate
random permutation) and subsequently convoluted with a 
(low-frequency) Gaussian, then the probability that the resulting
value at a location $m \in \mathcal{S}_k^c$ is of order
$\mathcal{O}(1)$ is small.  If $m \in \mathcal{S}_k$, the value
at $m$ is of order $\mathcal{O}(1)$ with probability $1$. This
realization allows us to develop an \emph{efficient statistical test}
to identify $\mathcal{S}_k$ from the knowledge of $\mathcal{M}_k$. The
process is shown schematically in Figure
\ref{1DsFFT:findsupport:schematic}. Specifically, the four following
steps are performed: 1) permute samples randomly, 2) apply a diagonal
Gaussian filter, 3) compute a small FFT, 4) eliminate elements that do
not belong to the aliased support. To help the reader better understand,
we once again proceed through an example. To begin with, assume
$M_k=40$, and  
 \begin{align*}
 \mathcal{S}_k = \{	1,  23, 35\},\, \mathcal{M}_k = \{ 1, 3, 15, 21, 23, 35 \},\, \mathcal{W}_k = \{ 1,2, ... , 39\}
\end{align*} 
as in step $k=3$ of the previous section (line A), Figure \ref{1DsFFT:findsupport:schematic}). The first step is to randomly
shuffle the elements of $ \mathcal{M}_k $ within $\mathcal{W}_k$ by
applying a permutation operator $\Pi_Q (\cdot)$ in sample 
space
\begin{equation}
\label{support:permutation}
\Pi_Q \left ( f_{n;\,M_k} \right ) =  f_{ (n [Q]^{-1}_{M_k} ) \, \mathrm{mod} \, M_k; \, M_k}
\end{equation}
for some integer $Q \in \mathcal{Q}(M_k)$ ($[Q]_{M_k}^{-1}$ being the unique inverse of $Q$ modulo
$M_k$) by Lemma \ref{1DsFFT:findsupport:hash} (Appendix \ref{appendix:support}). This
is equivalent to shuffling in frequency space as:  $j \rightarrow  (j Q) \, \mathrm{mod} M_k$. Indeed, after shuffling, evaluating at $(j Q) \, \mathrm{mod} M_k$ gives,
$$ \hat{f}^{(k)}_{(  ( (jQ) \, \mathrm{mod} M_k ) [Q]_{M_k}^{-1}) \, \mathrm{mod} M_k} = \hat{f}^{(k)}_j .$$ 
Furthermore, Lemma \ref{1DsFFT:findsupport:shuffle}
(Appendix~\ref{appendix:support}) shows that if $Q$ is chosen
uniformly at random within $\mathcal{Q}(M_k)$, the mapped elements of
the candidate support $\mathcal{M}_k$ will be more or less uniformly distributed
within the working support $\mathcal{W}_k$,  
\begin{equation*}
\mathbb{P} \left ( \left . \left | (i Q) \, \mathrm{mod} M_k - (j Q) \, \mathrm{mod} M_k  \, \right | \leq C \, \right | \,  i \not = j  \right ) \leq \mathcal{O}\left( \frac{C}{M_k} \right ).
\end{equation*}
For our example, assume $Q=13$. The
sets $\mathcal{S}_k$ and $\mathcal{M}_k$ are then mapped by $\Pi_Q(\cdot)$
to (line {\bf B)}, Figure \ref{1DsFFT:findsupport:schematic}). 
\begin{align*}
\mathcal{S}_k^{\mathrm{shuffled}} &=\{ (1\cdot13) \, \mathrm{mod} \, 40, (23 \cdot 13)\, \mathrm{mod} \, 40,\, (35 \cdot 13)\, \mathrm{mod} \, 40 \} = \{ 13,15,19 \} \\
\mathcal{M}_k^{\mathrm{shuffled}} &=\{ (1\cdot13) \, \mathrm{mod} \, 40,\,  (3\cdot 13)\, \mathrm{mod} \, 40,\, (15\cdot 13)\, \mathrm{mod} \, 40,\, (21 \cdot 13)\, \mathrm{mod} \, 40,\, (23 \cdot 13)\, \mathrm{mod} \, 40,\, (35 \cdot 13)\, \mathrm{mod} \, 40 \} \\
&= \{ 13,15,19,33,35,39 \} 
\end{align*}
\begin{figure}[H]
\begin{tikzpicture}

\draw[line width = 1.pt] (0,0) -- (\L,0);
\renewcommand*{\k}{1};
\foreach \x in \LocListSearchc{
	\draw[line width = 0.5pt] (\DELTA*\x,-0.1) -- (\DELTA*\x, 0.1);
	\draw (\DELTA*\x,0.3) node[scale=0.8] {\x};
	\k = \k+1;
	}
	
\foreach \x in \LocList{
	\draw[line width = 2.pt] (\DELTA*\x,-0.1) -- (\DELTA*\x, 0.1);
}

\draw (-1,0.) node {\bf A)};

\end{tikzpicture}

\begin{tikzpicture}
\draw[line width = 1.pt] (0,0) -- (\L,0);
\foreach \x in \LocListSearchShuffled{
	\draw[line width = 0.5pt] (\DELTA*\x,-0.1) -- (\DELTA*\x, 0.1);
	\draw (\DELTA*\x,0.3) node[scale=0.8] {\x};
	}
	
\renewcommand*{\k}{1};
\foreach \x in \LocListAliasShuffled{
	\draw[line width = 2.pt] (\DELTA*\x,-0.1) -- (\DELTA*\x, 0.1);
	\draw (\DELTA*\x,0.3) node[scale=0.8] {\x};
	\k = \k+1;
}
\draw (-1,0.) node {\bf B)};

\end{tikzpicture}

\begin{tikzpicture}
\draw[line width = 1.pt] (0,0) -- (\L,0);
\foreach \x in \LocListReduced{
	\draw (4*\DELTA*\x,0) node[cross=3pt,black]{};
	}
	
\draw[line width = 1.pt] (0,0) -- (\L,0);
\foreach \x in \LocListSearchShuffled {
	\draw[line width = 0.5pt] (\DELTA*\x,-0.1) -- (\DELTA*\x, 0.1);
	\draw (\DELTA*\x,0.3) node[scale=0.8] {\x};
	}

\renewcommand*{\k}{1};
\foreach \x in \LocListAliasShuffled{
	\draw[line width = 2.pt] (\DELTA*\x,-0.1) -- (\DELTA*\x, 0.1);
	\draw (\DELTA*\x,0.3) node[scale=0.8] {\x};
	\k = \k+1;
}

\draw[scale=0.25,thick,domain=0:40,samples=100, smooth,variable=\z,black] plot ({\z},{3*exp(-(\z - 13)*(\z-13)/9) + 1.5*exp(-(\z - 15)*(\z-15)/9) + 2.5*exp(-(\z - 19)*(\z-19)/9)});
\draw (-1,0.) node {\bf C)};

\end{tikzpicture}

\caption{\footnotesize Finding the aliased support $\mathcal{S}_k$ from
  knowledge of $\mathcal{M}_k$ (line {\bf A)}). First, indices are
  shuffled in sample space leading to a shuffling in frequency space
  (line {\bf B)}). A Gaussian filter is applied followed by a small
  FFT (line {\bf C)}) on a grid $G$ ($\times$). The points of $\mathcal{M}_k$ for
  which the value of the result of the last step at their closest
  neighbor in $G$ is small are discarded leaving only the aliased
  support $\mathcal{S}_k$. 
} 
\label{1DsFFT:findsupport:schematic}
\end{figure}
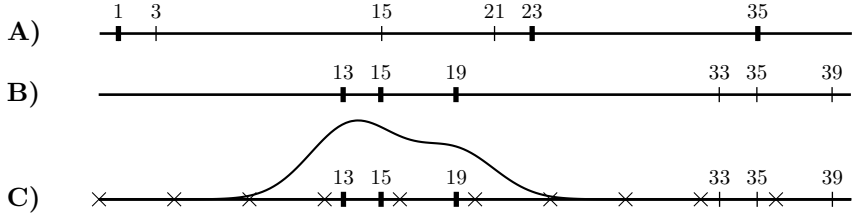
This step is followed by the application of a diagonal Gaussian
filtering operator $\Psi_\sigma(\cdot)$ having elements 
\begin{equation}
\label{support:filterdef}
 g_\sigma \left ( \frac{m}{M_k} \right ) = \sqrt{\pi} \sigma \sum_{h \in  \mathbb{Z}} e^{- \pi^2 \sigma^2 \left ( \frac{m+hM}{M}  \right )^2} 
\end{equation}
in sample space (step 2). By the properties of the Fourier transform,
this is equivalent to a convolution in frequency space (line  {\bf
  C)}, Figure \ref{1DsFFT:findsupport:schematic}), implying the
equality:  
\begin{equation}
\label{support:smoothFT}
\left [ \Psi_\sigma \left ( \Pi_Q \left ( f_{n;\,M_k} \right ) \right ) \right ] (\xi)  = \mathcal{F}^* \left [ \sum_{j \in\mathcal{S}_k} \hat{f}^{(k)}_{ j} \,  e^{- \frac{ \left | x  -   ( j Q ) \, \mathrm{mod} \, M_k \right  |^2}{\sigma^2} } \right ] (\xi). 
\end{equation}
The function is now \emph{bandlimited} (with bandwidth of order
$\mathcal{O}(K)$ thanks to our choice for $\sigma$; Algorithm
\ref{RASR:recoveraliasedsupp}), so this expression can be discretized (samples $\times$, Figure
\ref{1DsFFT:findsupport:schematic}) to produce our main expression,
\begin{equation}
\label{RASR:maineq}
 \phi^{(k)}_n (Q) = F_{\mathcal{A}(K,M_k)} \left [ \Psi_\sigma \left ( \Pi_Q \left ( f_{n;\,M_k} \right ) \right ) \right ]_n = \frac{1}{M_k} \sum_{m \in \mathcal{A}(K,M_k)}  e^{2 \pi i \frac{ n  m }{K} }  \, g_\sigma \left ( \frac{m}{M_k} \right )\,  f_{ (m [Q]^{-1}_{M_k} ) \, \mathrm{mod} \, M_k; \, M_k}
\end{equation}

In particular, we note that if $n$ if of the form $ j \frac{M_k}{K}$ for
$j=0,...,K-1$, the last step can be performed through a small size-$K$
FFT. This corresponds to step 3 of the aliased support
recovery algorithm.  The knowledge of,
$$\left \{ \phi^{(k)}_{j
  \frac{M_k}{K}} \right \}_{j=0}^{K-1}$$
can be used to recover $\mathcal{S}_k$ from $\mathcal{M}_k$ rapidly, seen
intuitively as follows: by construction, $\phi^{(k)}_n$ can be ``large'' only if the distance between $n$ and some element belonging to
the shuffled and aliased support, i.e., some element of $\left \{ (lQ) \, \mathrm{mod} M_k\right \}_{l \in \mathcal{S}_k}$, is
smaller than $\mathcal{O}(\sigma)$, which in turn occurs only if
the distance between $ \left [ n \frac{K}{M_k}  \right ] \,
\frac{M_k}{K}$ and the shuffled elements of the aliased support is smaller than
$\mathcal{O}(\sigma)$ as well (by the triangle inequality). However,
because of the randomness introduced by the shuffling, and because of
the particular choice of $\sigma$, it can be shown (Proposition
\ref{1DsFFT:findsupport:findsupport_proof}) that for any fixed $n \in
\mathcal{M}_k$, the probability that a computed element $ \phi^{(k)}_{
  \left [  (n Q \, \mathrm{mod} \, M_k) \frac{K}{M_k} \right ] \frac{M_k}{K} }$ is ``large'' for
multiple independent trials is small if $n \in \mathcal{M}_k \cap
\mathcal{S}_k^c$ and equal to $1$ if $n \in \mathcal{M}_k \cap
\mathcal{S}_k$. This fact allows for the construction of an efficient
statistical test based on the knowledge of the quantities found in
Eq.\eqref{RASR:maineq} to discriminate between the points of
$\mathcal{M}_k \cap \mathcal{S}_k^c$ and those of $\mathcal{M}_k \cap
\mathcal{S}_k$ (step 4). Such a test constitutes the core of Algorithm
\ref{RASR:recoveraliasedsupp}, and its correctness follows from the
following proposition (Appendix~\ref{appendix:support}).

As for the computational cost, the permutation and filtering
(multiplication) steps 
(1 and 2) both incur a cost of $\tilde{\mathcal{O}}(R \log(R))$ since
only the samples for which the filter is of order $\mathcal{O}(1)$ are
considered (and there are $\mathcal{O}( K) = \tilde{\mathcal{O}}(R \log(R))$ of them following our choice of $\sigma$ and
$K$). These are followed by an FFT (step 3) of size $ \mathcal{O}(K
)$ which carries a cost of order  $\tilde{\mathcal{O}}(R \log^2(R))$. Finally, step 4 involves checking a simple property on each of the
$M_k$ elements of $\mathcal{M}_k$ incurring a cost of
$\tilde{\mathcal{O}}(R \log(R))$.  This is repeated $\mathcal{O}(\log(p))$ times
for a probability $(1-p)$ of success. Thus, extracting $\mathcal{S}_k$ from $\mathcal{M}_k$ requires merely
$\tilde{\mathcal{O}}(\log(p)  R \log(R) ) $ samples and
$\tilde{\mathcal{O}}(\log(p) R \log^2(R))$ 
computational time for fixed $p$, as claimed.  
\begin{prop}
\label{RASR:correctness}
In the noiseless case, Algorithm \ref{RASR:recoveraliasedsupp} outputs
$\mathcal{S}_k$, the aliased support of the vector $\hat{f}$ at step
$k$, with probability at least $(1-p)$ using $\tilde{\mathcal{O}}(\log(p)  R \log(R) ) $ samples and $\tilde{\mathcal{O}}(\log(p) R \log^2(R))$ computational steps. 
\end{prop}
\begin{proof}
Refer to Algorithm \ref{RASR:recoveraliasedsupp} and Proposition
\ref{1DsFFT:findsupport:findsupport_proof} as well as the above discussion. 
\end{proof}

\begin{algorithm}[H]
\caption{ $\mathrm{FIND\_ALIASED\_SUPPORT}(\mathcal{M}_{k}, M_k, K,\alpha,p, \delta, \mu, \Delta )$}
\label{RASR:recoveraliasedsupp}
\begin{algorithmic}[1]
\State Let $\sigma = \frac{\alpha \, \frac{M_k}{2R} }{\sqrt{  \log\left ( \frac{2 R \Delta}{\delta} \right )}} $ and $L = \log_\alpha (p) $.
\State $\mathcal{S}_k \leftarrow \mathcal{M}_k$
\For{ $l$ {\bf from} $1$ {\bf to} $L$}
	\State	Pick $Q^{(l)} \in \mathcal{Q}(M_k)$ uniformly at random.
	\State Compute: $\phi^{(k)}_{j \frac{M_k}{K} } ( Q^{(l)} ), \; j = 0, 1, ..., K-1$ (Eq.\eqref{RASR:maineq}). 
	 	\For{ $n \in \mathcal{M}_k$  }
			\If{ $\left |\phi^{(k)}_{ \left [ \left ( n Q^{(l)} \, \mathrm{mod} \, M_k    \right )\frac{K}{M_k}  \right ]\, \frac{M_k}{K} }( Q^{(l)} )  \right  |  < \frac{\delta \, \mu}{2} $ }
				\State Remove $j$ from $\mathcal{S}_k$.
			\EndIf
		\EndFor
\EndFor
\State Output: $\mathcal{S}_k$.
\end{algorithmic}
\end{algorithm}

\subsection{Recovering values from knowledge of the support}
\label{1DsFFT:CV}
In this section, assume a set size $\mathcal{O}(K)$ containing the support
$\mathcal{S}$ has been recovered.  We now show how the
values of the nonzero Fourier coefficients 
of $\hat{f}$ 
in Eq.~\eqref{1DsFFT:orig_func} can be rapidly computed using this information. For
this purpose, assume $f(x)$ can be sampled at locations: $\left \{
\frac{q \, \mathrm{mod} \, P^{(t)}}{ P^{(t)}}   \right
\}_{q=0}^{ P^{(t)}-1}$ for $t = 0 , 1, ..., T$, and $\{ P^{(t)}
\}_{t=1}^T $ some random prime numbers on the order of $\mathcal{O}(R
\log_R(N) )$ (see Algorithm \ref{1DsFFT:compute_values}). It follows
that 
\begin{equation}
\label{1DsFFT:CV:system_eqn}
f^{(t)}_{q  \, \mathrm{mod} \, P^{(t)}; \, P^{(t)}} = \sum_{j
  \in \mathcal{S}} e^{-2\pi i \,\frac{q \, ( j \,
    \mathrm{mod} \, P^{(t)} ) }{P^{(t)}} }\,\hat{f}_{ j }  =
\sum_{l=0}^{P^{(t)} -1} e^{-2\pi i \,\frac{q\,l}{P^{(t)}} } \, \left (
\sum_{j \in \mathcal{S}:  j  \,
  \mathrm{mod} \, P^{(t)} = l }  \hat{f}_j  \right ) 
\end{equation}
for $t = 0,1, ..., T$. The outer sum is seen to be a DFT of size $
P^{(t)}$ of a shuffled and aliased vector, whereas the inner sum can
be expressed as the application of a binary matrix $B^{(t)}_{q,j}$
with entries 
\begin{equation*}
B^{(t)}_{q,j} =
\left\{
	\begin{array}{ll}
		1  & \mbox{if } \, j  \, \mathrm{mod}\, P^{(t)} = q   \\
		0 & \mbox{else}
	\end{array}
\right.
\end{equation*}
to the vector with entries' index corresponding to those of the support of
$\hat{f}$. In particular, each such matrix is sparse with exactly $\#
\mathcal{S} = \mathcal{O}(R)$ nonzero
entries. Eq.~\eqref{1DsFFT:CV:system_eqn} can further be written in
matrix form as 
\begin{equation}
\label{1DsFFT:CV:system}
 [FB] \, \hat{f} =  \left[ \begin{array}{cccc}
F^{(1)} & 0  & ...  & 0   \\
0 & F^{(2)} & ... & 0  \\
... & ... & ... & ... \\
0 & 0 & ... & F^{(T)}   \end{array} \right] \,
 \left[ \begin{array}{c}
B^{(1)}   \\
B^{(2)}  \\
...  \\
B^{(T)}  \end{array} \right]   \, \hat{f} = 
 \left[ \begin{array}{c}
f^{(1)}   \\
f^{(2)}  \\
...  \\
f^{(T)}  \end{array} \right]  = f_0,
\end{equation}
where $F^{(t)}$ is a standard DFT matrix of size $P^{(t)}$. 
Proposition \ref{1DsFFT:CV:Neumann} states that if $T = \mathcal{O}(1)$ is sufficiently large, then with nonzero probability $\frac{1}{T}
(FB)^* (FB) = I + \mathcal{P}$, where $I$ is the identity and
$\mathcal{P}$ is a perturbation with $2$-norm smaller than
$\frac{1}{2}$. When this occurs, one can solve the linear system
through the Neumann series, i.e.,
$$\hat{f} = \sum_{n=0}^\infty  (I - B^* B )^n\, (FB)^* f_0$$
This constitutes the core of Algorithm \ref{1DsFFT:compute_values}. The correctness of the algorithm is provided in Proposition~\ref{1DsFFT:CV:alg3_correctness}.

Since each matrix $B^{(t)}$ contains exactly $R$
nonzero entries, both $B$ and $B^* B$ can be applied in order $RT =
\mathcal{O}(R \log_R (N))$ steps. In 
addition, since $F$ is a block diagonal matrix with $T = \mathcal{O}(1)$ blocks consisting of DFT matrices of size $\mathcal{O}(R\log_R(N))$, it can be applied in order $\tilde{\mathcal{O}}(R \log (N))$ thanks to the FFT. Finally, for an accuracy $\eta$ the Neumann series can be truncated after $\mathcal{O}(\log(\eta))$ terms, and the process needs to be repeated at most $\log(p)$ times for a probability $p$ of success. Therefore, the cost of
computing the nonzero values of $\hat{f}$ is bounded by
$\tilde{\mathcal{O}}\left ( (\log(p) + \log(\eta)) \, R \log (N)  \right )$ and uses at most $\mathcal{O}\left ( (\log(p) + \log(\eta)) R \log_R (N) \right )$ samples as claimed.

\begin{restatable}{prop}{ValRecProp}
\label{1DsFFT:CV:alg3_correctness}
Assume the support $\mathcal{S}$ of $\hat{f}$ is known. Then Algorithm
\ref{1DsFFT:compute_values} outputs an approximation to the nonzero
elements of $\hat{f}$ with error bounded by $\eta$ in the
$\ell^2$-norm, with probability greater than or equal to $1-p$ using  $\mathcal{O}\left ( (\log(p) + \log(\eta)) R \log_R (N) \right )$ samples and $\tilde{\mathcal{O}}\left ( (\log(p) + \log(\eta)) \, R \log (N)  \right )$ computational steps. 
\end{restatable}

\begin{algorithm}[H]
\caption{$\mathrm{COMPUTE\_\,VALUES}(\mathcal{S},R,N,p,\mu,\Delta,\eta)$}
\label{1DsFFT:compute_values}
\begin{algorithmic}[1]
\State Let $T =  4 $, $Z = \lceil \log_{\frac{1}{2} } (\eta) \rceil  $ and $L = \lceil \log_{\frac{1}{2}} (p) \rceil$.
\For{ $t$ {\bf from} $1$ {\bf to} $L$}
	\State Pick $\{ P^{(t)} \}_{t=1}^{T}$ i.i.d. uniform r.v. chosen among the set containing the smallest $ 4R \log_R(N)$ prime numbers greater than $R$.
	\State Sample $\left \{ f_{ n  \, \mathrm{mod} \, P^{(t)} ; P^{(t)}} \right  \}_{n=0}^{P^{(t)}-1}, \; t = 1,..., T$
	\State Compute $\hat{f}_0 \leftarrow (FB)^* f_0$ (Eq.\eqref{1DsFFT:CV:system_eqn}).
	\If{ $|| (B^* B ) \, \hat{f}_0 ||_2 < \frac{1}{2}$ }
		\State 	$ \hat{f}  \leftarrow \sum_{n=0}^{Z} (I-B^* B )^n \hat{f}_0$ 
		\State Output: $\hat{f}$.
		\State Exit.
	\EndIf
\EndFor
\end{algorithmic}
\end{algorithm}
%
%

\subsection{Stability to low-level noise}
\label{stability}

As discussed previously, the theory underlying the algorithms
introduced in Section~\ref{1DsFFT} has been designed for vectors which
are exactly sparse. In this section, we discuss the effect of low-level noise. In fact, we show that if the sparse vector
of Fourier coefficients takes the form $ \hat{f} + \hat{\nu} $, where
$\sqrt{N} || \hat{\nu} ||_{2} < \eta$ for some ``small''
$\eta$, the sMFFT algorithm recovers the support and values of
$\hat{f}$ with the same guarantees as described earlier.  

\subsubsection{Support recovery} 
The most important quantity for the fast recovery of
the support is Eq.\eqref{RASR:maineq}, so in the 
presence of noise, 
\begin{equation}
\label{noiseMainEq}
 \phi^{(k)}_n (Q) =  \frac{1}{M_k} \sum_{m \in \mathcal{A}(K,M_k)}  e^{2 \pi i \frac{ n  m }{K} }  \, g_\sigma \left ( \frac{m}{M_k} \right )\,  \left ( f_{ (m [Q]^{-1}_{M_k}  ) \, \mathrm{mod} \, M_k; \, M_k}  + \nu_{ (m [Q]^{-1}_{M_k} ) \, \mathrm{mod} \, M_k; \, M_k}  \right ).
\end{equation}
The second term in this expression is the error term and can be
\emph{uniformly bounded} by the following lemma:
\begin{restatable}{lemma}{stabilitysuppthm}
\label{stability:noisesupport}
Assuming the noise term $\hat{\nu}$ is such that $||\hat{\nu} ||_{2} <
\frac{\eta}{\sqrt{N}} $, the error term of the computed value in
Eq.\eqref{noiseMainEq} is uniformly bounded by 
\begin{equation*}
\left | \left |  \psi^{(k)}_n (Q)   \right | \right |_\infty = \left | \left |   \frac{1}{M_k} \sum_{m \in \mathcal{A}(K,M_k)}  e^{2 \pi i \frac{ n  m }{K} }  \, g_\sigma \left ( \frac{m}{M_k} \right )\,  \nu_{ (m [Q]^{-1}_{M_k} ) \, \mathrm{mod} \, M_k; \, M_k}    \right | \right |_\infty  <  \mathcal{O}(\eta).
\end{equation*}
\end{restatable}
Algorithm \ref{1DsFFT:findsupport_alg} tests whether $\left |
\phi^{(k)}_{\left [ (i Q^{(l)} \, \mathrm{mod} M_k)\,  \frac{K}{M_k} \right ] \frac{M_k}{K} }  (Q^{(l)}
)\right  |  > \delta \, \mu $ in order to discriminate between elements
of the candidate and aliased supports. The presence of noise can skew
this test in two ways: 1) by bringing the computed value below the
threshold when $i \in \mathcal{S}_k$ or 2) by bringing the value above
the threshold multiple times when $i \not \in \mathcal{S}_k$. Either
way, if $\eta$ is small enough, i.e., such that $\left | \left |   \psi^{(k)}_n
(Q^{(l)})     \right | \right |_\infty \leq \frac{\delta \mu}{2}$,
it can be shown that the conclusion of
Proposition \ref{1DsFFT:findsupport:findsupport_proof} follows through
with similar estimate, by simply replacing $\delta  $ with
$\frac{\delta  }{2}$ in the proof.  

\subsubsection{Recovering values from knowledge of the support} 
It is quickly
observed that the recovery of the values is a 
well-conditioned problem. Indeed, since $\frac{1}{T} (FB)^* (FB) = I
- \mathcal{P}$, and $ || \mathcal{P} ||_2 \leq \frac{1}{2} $ with high 
probability by Proposition \ref{1DsFFT:CV:Neumann}, a simple argument
based on the singular value decomposition produces the following corollary,
\begin{cor}
Under the hypothesis of Proposition~\ref{1DsFFT:CV:Neumann}, $\left (
  \frac{1}{T} (FB)^* (FB) \right )^{-1}$ exists, and $\left | \left | \left (\frac{1}{T} (FB)^* (FB) \right )^{-1} \right | \right |_2 \leq 2$
with probability greater than or equal to $\frac{1}{2}$.
\end{cor}
Therefore, the output of Algorithm~\ref{1DsFFT:compute_values} is such
that
\begin{equation*}
|| \hat{f}^{\mathrm{sMFFT}} - \hat{f} ||_2 \leq \left | \left |  \left
( \frac{1}{T} (FB)^* (FB) \right )^{-1} \right | \right |_{2} \, ||
(FB)^* \nu ||_{2} \leq 2 ||B||_2  \, || \nu ||_{2} =
\mathcal{O}(\eta).
\end{equation*}
This, together with Proposition \ref{1DsFFT:CV:alg3_correctness}, demonstrates the stability of Algorithm \ref{1DsFFT:compute_values} in the noisy case.

\section{The multi-dimensional sparse FFT}
\label{MD}
Whenever dealing with the multidimensional DFT/FFT, it is
assumed that  the function of interest is both periodic and
bandlimited with fundamental period $[0,1)^d$, i.e., 
$$f(x) = \sum_{j
    \in ( [0,M) \cap \mathbb{Z} )^d }   e^{-2\pi i \, x \cdot j} \,
    \hat{f}_j$$
 for some finite $M \in \mathbb{N}$ and $j \in  \mathbb{Z}^d , $ up to some rescaling. Computing 
the Fourier coefficients is then equivalent to computing the
$d$-dimensional integrals,
 $$\hat{f}_n = \int_{[0,1]^d} e^{-2\pi i \, n
  \cdot x} \, f(x) \, \mathrm{d} x,$$
  and this is traditionally
achieved through a ``dimension-by-dimension'' trapezoid rule
\cite{van1992computational,chu1999inside} 
\begin{equation}
\label{MD:funcMDperiodic}
\hat{f}_{(j_1, j_2, ..., j_d)} =  \sum_{n_1 = 0}^{M - 1} \frac{e^{2\pi i \,
  \frac{j_1 n_1}{M}}}{M} \left ( ... \left ( \sum_{n_{d-1} = 0}^{M - 1}
    \frac{e^{2\pi i \, \frac{j_{d-1} n_{d-1}}{M}}}{M} \,  \left ( \sum_{n_d =
        0}^{M - 1} \frac{e^{2\pi i \, \frac{j_d n_d}{M}}}{M}  f_{(n_1, n_2, ...,
        n_d)}  \right )   \right )  \right ) .
\end{equation}
However, Proposition~\ref{MD:rank1lemma} below shows that it is also possible to
re-write the $d$-dimensional DFT as that of a 1D function
with Fourier coefficients equal to those of the original function,
but with different ordering.
\begin{restatable}{prop}{mdsfft}
\label{MD:rank1lemma}{(Rank-1 $d$-dimensional DFT)}
Assume the function $f:[0,1)^d \rightarrow \mathbb{C}$ has form
  \eqref{MD:funcMDperiodic}. Then, 
\begin{equation}
\label{MD:rank1}
 \int_{[0,1]^d} e^{-2\pi i \, j \cdot x }\, f(x) \, \mathrm{d} x = \frac{1}{N} \sum_{n=0 }^{N-1} e^{-2\pi i \, j \cdot x_n} \, f(x_n) 
\end{equation}
for all $j \in  [0,M)^d \cap \mathbb{Z}^d $, where $x_n = \frac{n g \, \mathrm{mod} N}{N}$, $g = (1, M, M^2, ..., M^{d-1})$ and $ N = M^d$.
\end{restatable}
Now, the right-hand side of Eq.~\eqref{MD:rank1} can be written in two different ways (due to
periodicity); namely,  
\begin{equation*}
 \frac{1}{N} \sum_{n=0 }^{N-1} e^{-2\pi i \, j \cdot \frac{n g \,
     \mathrm{mod} N}{N} } \, f\left ( \frac{n g \, \mathrm{mod} N}{N}
 \right )  =  \frac{1}{N} \sum_{n } e^{-2\pi i \, \frac{ (j\cdot g) n 
   }{N} } \, f \left ( \frac{n g}{N}  \right )  
\end{equation*}

Geometrically, the left-hand side represents a quadrature rule with
points $x_n = 
\frac{n g \, \mathrm{mod} N}{N} $ distributed (more-or-less uniformly)
in $[0,1)^d $ (Figure~\ref{MD:Rank1Interpretation}, left; grey dots).
  The right-hand side represents an equivalent quadrature where the
  points $x_n = \frac{n g}{N} $ now lie on a \emph{line} embedded in
  $\mathbb{R}^d$ (Figure~\ref{MD:Rank1Interpretation}, right; grey
  dots). The location at which the lattice (thin black lines)
  intersects represents the standard multidimensional DFT
  samples. 
\begin{figure}[h]
\begin{center} 

\begin{tikzpicture}[scale = 0.725]

\begin{scope}[local bounding box=aa,rotate=-90]
\foreach \x in {0,...,\NBigCellhoriz}{
	\draw[line width = 2.5pt] (\BigLat*\x,0) -- (\BigLat*\x, \NBigCellhoriz*\BigLat);
	}
\foreach \x in {0,...,\NBigCellhoriz}{
	\draw[line width = 2.5pt] (0,\BigLat*\x) -- ( \NBigCellhoriz*\BigLat, \BigLat*\x);
	}
	
\foreach \x in {0,...,\NSmCellhoriz}{
	\draw[line width = 1pt] (\SmLat*\x,0) -- (\SmLat*\x, \NSmCellhoriz*\SmLat);
	}

\foreach \x in {0,...,\NSmCellhoriz}{
	\draw[line width = 1pt] (0,\SmLat*\x) -- ( \NSmCellhoriz*\SmLat, \SmLat*\x);
	}

\foreach \x in {0,...,\NSmCellminus}{	
	\draw[line width = 1.5pt, color=gray] (\x*\SmLat ,0) -- (\x*\SmLat +\SmLat ,\BigLat);
	\foreach \y in {0,...,\NSmCell}{
		\draw[fill=gray] (\x*\SmLat + \y/\NSmCellsq*\BigLat, \y/\NSmCell*\BigLat) circle  (2pt);
	}
}
\end{scope}

\begin{scope}[shift={(6,0)},local bounding box=bb,rotate=-90]
\foreach \x in {0,...,\NBigCellhoriz}{
	\draw[line width = 2.5pt] (\BigLat*\x,0) -- (\BigLat*\x, \NBigCellvert*\BigLat);
	}
\foreach \x in {0,...,\NBigCellvert}{
	\draw[line width = 2.5pt] (0,\BigLat*\x) -- ( \NBigCellhoriz*\BigLat, \BigLat*\x);
	}
	
\foreach \x in {0,...,\NSmCellhoriz}{
	\draw[line width = 1pt] (\SmLat*\x,0) -- (\SmLat*\x, \NSmCellvert*\SmLat);
	}

\foreach \x in {0,...,\NSmCellvert}{
	\draw[line width = 1pt] (0,\SmLat*\x) -- ( \NSmCellhoriz*\SmLat, \SmLat*\x);
	}
	
\draw[line width = 1.5pt, color=gray] (0,0) -- (\NBigCellhoriz*\BigLat,\NBigCellvert*\BigLat);
\foreach \x in {0,...,\NSmCellsq}{
	\draw[fill=gray] (\x/\NSmCellsq*\BigLat, \x/\NSmCell*\BigLat) circle  (2pt);
}

\end{scope}

\end{tikzpicture}

\end{center}
\caption{\footnotesize Geometric interpretation of rank-1 $d$-dimensional DFT in
  2D. The thick black box represents fundamental periodic domain. The grey dots
  represent rank-1 discretization points. The 2D grid represents standard discretization points. Left: the rank-1 $d$-dimensional quadrature
  interpreted as a 2D discretization over the fundamental periodic
  region. Right: the rank-1 $d$-dimensional quadrature interpreted as a uniform
  discretization over a line in $\mathbb{R}^2$. } 
\label{MD:Rank1Interpretation}
\end{figure}
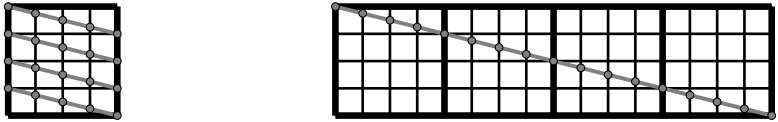
In short, Proposition \ref{MD:rank1lemma} allows one to write \emph{any} $d$-dimensional DFT as a one-dimensional DFT by picking the appropriate sample points (Proposition \ref{MD:rank1lemma}) and proceeding to a re-ordering of the Fourier coefficients through the isomorphism 
\begin{equation*}
  \tilde{n}  : \left \{ n \in ([0,M) \cap \mathbb{Z})^D \right \}
  \rightarrow  n\cdot g = n_0 + n_1\, M + ... + n_{D-1} \, M^{D-1} \in [0, M^D)
  \cap \mathbb{Z}^D. 
\end{equation*}
We use this sort of sampling to treat of the multidimensional problem, which we always convert to a 1D problem. Then, we address the 1D problem through the algorithm presented in the previous sections, and finally map the result back to its original $d$-dimensional space.\\

{\bf Remark. } In general, there is no need for the signal's spectrum to reside on a \emph{regular grid}. Indeed if the signal is sparse in the continuous domain, one can approximate such function on the hypercube by a function which spectrum does lie on a regular grid, and which spectral support correspond to the closest neighbor of the original spectral support within the regular grid. For some error $\epsilon$, a grid with spacing $\mathcal{O}\left( \frac{1}{\epsilon} \right )$ containing $N = \mathcal{O}\left( \frac{1}{\epsilon^d} \right )$ unknowns should suffice.

\section{Numerical results}
\label{Numer}
We have implemented our sMFFT algorithm in
\texttt{MATLAB}\footnote{\texttt{MATLAB} is a trademark of
  Mathworks} and present a few numerical results which exhibit the
claimed 
scaling. All simulations were carried out on a small cluster possessing 4
Intel® Xeon® E7-4860 v2 processors and 256GB of RAM, with the
\texttt{MATLAB}
flag $\mathtt{-singleCompThread}$  to ensure fairness through the use
of a single computational thread. The numerical experiments presented
here fall in two categories: 1) dependence of running time as a
function of the total number of unknowns $N$ for a fixed number of
nonzero frequencies $R$, and 2) dependence of running time as a
function of the number of nonzero frequencies $R$ for a fixed total
number of unknowns $N$. All experiments were carried out in three
dimensions (3D) with additive Gaussian noise with variance $\eta$. The
nonzero 
values of $\hat{f}$ were picked randomly and uniformly at random in $[0.5,
  1.5]$, and the remaining parameters were set according
to Table \ref{results:parameters}. All comparisons are perdormed with
the \texttt{MATLAB} $\mathtt{fftn}(\cdot)$ function, which 
uses a dimension-wise decomposition of the DFT (see Section \ref{MD})
and a 1D FFT routine along each dimension. 
\begin{table}[h]
\scriptsize
\caption{\scriptsize Values of parameters required by Algorithm
  \ref{1DsFFT:algorithm}-\ref{1DsFFT:compute_values} and used for
  numerical experiments} 
\begin{center}
\begin{tabular}{|c|c|c|c|} \hline
  Parameter & Description & Value (Case 1) & Value (Case 2) \\
  \hline
$N$ & Total number of unknowns &variable &  $10^8$\\
$R$ & Number of nonzero frequencies &50 & variable \\
$\alpha$ & Gaussian filter parameter &$0.15$ &  0.15 \\
$\delta$ & Statistical test parameter &$0.1$&  0.1 \\
$p$ & Probability of failure & $10^{-4}$ & $10^{-4}$ \\
$d$ & Ambient dimension & $3$ &  3 \\
$\eta$ & Noise level & $10^{-2}$ &  $10^{-2}$ \\ \hline
\end{tabular}
\end{center}
\label{results:parameters}
\end{table}%
For case 1), we picked $R=50$ nonzero frequencies distributed
uniformly at random on a 3D lattice having $N^{1/3}$ elements in each
dimension for different values of $N\in [10^3, 10^{10}]$. The results
are shown in Figure \ref{results:N_dependence} (left). As can be
observed, the cost of computing the DFT through the sMFFT remains more
or less constant with $N$, whereas that the the \texttt{MATLAB}
$\mathtt{fftn}(\cdot)$ function increases linearly. This is the
expected behavior and demonstrates the advantages of the sMFFT over
the FFT. Also note that the largest relative $\ell^2$-error observed
was $9.3\cdot 10^{-3}$ which is on the order of the noise level, as
predicted by the theory. 

 \begin{figure}[H]
 \begin{center}
 	\pgfplotsset{every axis legend/.append style={at={(0.32,0.95)},anchor=north}}
 	\pgfplotsset{every axis/.append style={thick}}
 	 \begin{tikzpicture}[scale = 0.6]

	\begin{axis}[
		ylabel=Running time - $\log_{10}\left( t\right)$,
		xlabel=Number of unknowns - $\log_{10}\left( N \right)$]
		
	\addplot[color=black,mark=o] coordinates {
(3.6124,-3.2)
(4.5154,-2.82)
(5.4185,-1.78)
(6.3216,-0.88)
(7.2247,0.25)
(8.1278,1.29)
(9.0309,2.23)
	};

	\addplot[color=red,mark=x] coordinates {
(3.6124,-2.34) 
(4.5154,-2.30)
(5.4185,-2.38)
(6.3216,-2.11)
(7.2247,-2.13)
(8.1278,-2.06)
(9.0309,-1.96) 
	};

\legend{%
$\mathtt{fftn}(\cdot) \; \mathrm{(Matlab)}$\\
$\mathrm{sMFFT}$\\%
}

	\end{axis}
\end{tikzpicture}%
	  \begin{tikzpicture}[scale = 0.6]

	\begin{axis}[
		xlabel=Number of nonzero frequencies - $\log_{10}\left( R \right)$,
		ylabel=Running time - $\log_{10}\left( t \right)$]
		
	\addplot[color=black,mark=o] coordinates {
(0.6021,0.3160)
(0.9031,0.3160)
(1.2041,0.3160)
(1.5051,0.3160)
(1.8062,0.3160)
(2.1072,0.3160)
(2.4082,0.3160)
(2.7093,0.3160)
(3.0103,0.3160)
(3.3113,0.3160)
(3.6124,0.3160)
(3.9134,0.3160)
(4.2144,0.3160)
	};

	\addplot[color=red,mark=x] coordinates {
(0.6021,-1.8)
(0.9031,-1.94)
(1.2041,-1.82)
(1.5051,-1.89)
(1.8062,-1.77)
(2.1072,-1.38)
(2.4082,-1.17)
(2.7093,-1.08)
(3.0103,-0.91)
(3.3113,-0.64)
(3.6124,-0.44)
(3.9134,-0.17)
(4.2144,0.12)
	};

\legend{%
$\mathtt{fftn}(\cdot) \; \mathrm{(Matlab)}$\\
$\mathrm{sMFFT}$\\%
}

	\end{axis}
\end{tikzpicture}%
	
	
 	\end{center}
 	\caption{\footnotesize  Left: Running time vs number of unknowns
     ($N$) for the \texttt{MATLAB} $\mathtt{fftn}(\cdot)$ (black) and
     the sMFFT (red) in three dimensions (3D), with $R=50$ nonzeros
     and noise $\eta = 10^{-3}$. Right: Running time vs number of
     nonzero frequencies ($R$) for the \texttt{MATLAB}
     $\mathtt{fftn}(\cdot)$ (black) and the sMFFT (red) in three
     dimensions (3D) and for $N=10^8$ and noise $\eta = 10^{-3}$.} 
 	\label{results:N_dependence}
 \end{figure}
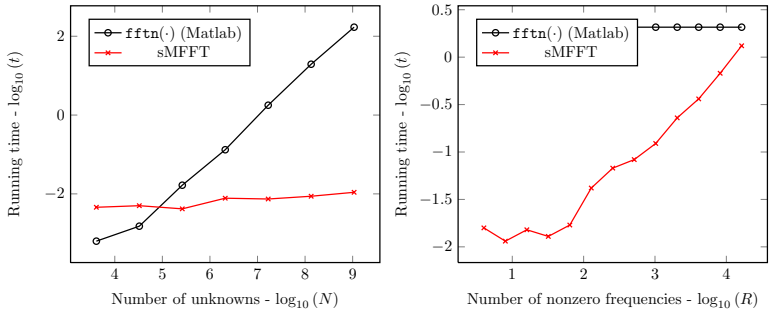
For case 2), we fixed $N =  \mathcal{O}(10^8)$ and proceeded to
compare the sMFFT algorithm with the \texttt{MATLAB}
$\mathtt{fftn}(\cdot)$ function as before (w/ parameters found in
Table \ref{results:parameters}). The results are shown in Figure 
\ref{results:N_dependence} (right). In this case, the theory states
that the 
sMFFT algorithm should scale quasi-linearly with the number of nonzero
frequencies $R$. A close look shows that it is indeed the case. For
this case, the largest relative $\ell^2$-error observed was $1.1\cdot
10^{-2}$, again on the order of the noise level and in agreement the
theory. Finally, the cost $\mathtt{fftn}(\cdot)$ function remains 
constant as the FFT scales like $\mathcal{O}(N\log(N))$ and is oblivious to
$R$.


%
%
%
%

 \section{Conclusion}
 
 We have introduced a sparse multidimensional FFT (sMFFT) for computing
 the DFT of a $N\times 1$ sparse, real-positive vector (having $R$
 nonzeros) that is stable to low-level noise, that exhibits a sampling complexity of ($\mathcal{O}(R \log(R) \log(N))$) and a
computational complexity ($\mathcal{O}(R \log^2(R)\log(N))$). The method also has a scaling of $\mathcal{O}(\log(p))$ with respect to probability of failure, $\mathcal{O}(\log(\eta))$ with respect to accuracy and $\mathcal{O}(1)$ with respect to dimension. We have provided a rigorous
 theoretical analysis of our approach demonstrating each
 claim. Finally, we have implemented our algorithm and provided
 numerical examples in 3D successfully demonstrating the claimed scaling. 
 
%

\appendix
\newpage
\label{appendix:genfunc}

\newpage
\section{Proofs}
\label{appendix:proofs}
In this appendix, we present all proofs and accompanying results
related to the statements presented in the main body of the work. 

\subsection{Proofs of Section \ref{1DsFFT}}
\label{appendix:support}

\begin{lemma}
\label{1DsFFT:findsupport:isomorphism}
Let $0 < Q \leq N \in \mathbb{N}$, $Q\perp N$. Then the map,
\begin{equation*}
 n \in \{0, 1, ..., N-1 \} \rightarrow n Q \, \mathrm{mod} \, N \subset \{0, 1, ..., N-1 \}
\end{equation*}
is an isomorphism.
\end{lemma}
\begin{proof}
Since the range is discrete and a subset of the domain, it suffices to show that the map is injective. Surjectivity will then follow from the pigeon hole principle. To show injectivity, consider $i,j \in  \{0, 1, ..., N-1 \}$, and assume,
\begin{equation*}
 i Q \, \mathrm{mod} \, N =  j Q \, \mathrm{mod} \, N
\end{equation*}
This implies (by definition) that there exists some integer $p$ such that,
\begin{equation*}
(i-j) Q = pN
\end{equation*}
so that $N$ divides $(i-j) Q$. However, $N \perp Q$ so $N$ must be a factor $(i-j)$. Now, $i,j$ are restricted to $ \{0, 1, ..., N-1 \}$ so,
\begin{equation*}
|i-j | < N,
\end{equation*}
and the only integer divisible by $N$ that satisfies this equation is $0$. Thus,
\begin{equation*}
i-j = 0 \Leftrightarrow i = j
\end{equation*}
which demonstrates injectivity.
\end{proof}

\begin{restatable}{lemma}{hashlemma}
\label{1DsFFT:findsupport:hash}
Let $ 0 < Q < M$ be an integer coprime to $M$ and,
\begin{equation*}
f_{n} =  \sum_{l=0}^{M-1} e^{-2 \pi i \, \frac{n \,  l}{M} }\,  \hat{f}_l
\end{equation*}
Then,
\begin{equation*}
\frac{1}{M} \, \sum_{n=0}^{M-1} e^{2 \pi i \, \frac{m \,  n}{M} }\, f_{ (n Q ) \, \mathrm{mod} \, M} =  \hat{f}_{(m [Q]^{-1}_M) \, \mathrm{mod} \, M }
\end{equation*}
where $0 < [Q]^{-1}_M  < M$ is the unique integer such that $[Q]^{-1}_M \, Q \, \mathrm{mod} \, M = 1 \, \mathrm{mod} M $. 
\end{restatable}
\begin{proof}
Consider 
\begin{align*}
\frac{1}{M} \, \sum_{n=0}^{M-1} e^{2 \pi i \, \frac{m \,  n}{M} }\, f_{(nQ) \, \mathrm{mod} M } &= \sum_{l=0}^{M-1}  \left (  \frac{1}{M}  \sum_{n=0}^{M-1} e^{2 \pi i \, \frac{n}{M} \left( m  - Q  l  \, \mathrm{mod} \, M \right ) }\right ) \,  \hat{f}_l \\
&=   \sum_{l=0}^{M-1}  \left (   \frac{1}{M}  \sum_{n=0}^{M-1} e^{2 \pi i \, \frac{n Q}{M} \left( m [Q]^{-1}_M   \, \mathrm{mod} \, M -   l \right ) }\right ) \,  \hat{f}_l
\end{align*}
However,
\begin{equation*}
 \frac{1}{M}  \sum_{n=0}^{M-1} e^{2 \pi i \, \frac{n Q}{M} \left( m[Q]^{-1}_M   \,
     \mathrm{mod} \, M -   l \right ) }  = \frac{1}{M}  \sum_{j=0}^{M-1} e^{2 \pi i \, \frac{j}{M} \left( m[Q]^{-1}_M   \,
     \mathrm{mod} \, M -   l \right ) } =  \delta_{ m[Q]^{-1}_M
   \, \mathrm{mod} \, M  , l}, 
\end{equation*}
where the second equality follows from the fact that $m \rightarrow m [Q]^{-1}_M \, \mathrm{mod} \, M $ is an isomorphism (Lemma \ref{1DsFFT:findsupport:isomorphism}). This implies that
\begin{equation*}
\frac{1}{M} \sum_{n=0}^{M-1} e^{2 \pi i \, \frac{m \,  n}{M} }\, f_{(nQ) \, \mathrm{mod} \, M} =  \hat{f}_{ (m[Q]^{-1}_M )  \, \mathrm{mod} \, M }
\end{equation*}
as claimed.
\end{proof}


\begin{restatable}{lemma}{shufflelemma}
\label{1DsFFT:findsupport:shuffle}
Let $M \in \mathbb{N} / \{0\}$ and
let $Q$ be a uniform random variable over $\mathcal{Q}(M)$ (Definition \ref{1DsFFT:Qset}). Then, 
\begin{equation*}
\mathbb{P} \left ( \left | j Q \, \mathrm{mod} M \right | \leq C  \right )  \leq \mathcal{O} \left( \frac{C}{M} \right )
\end{equation*}
for all $0 < j < M$ (up to a $\log(\log(M))$ factor).
\end{restatable}
\begin{proof}

Fix $0 < j, k <M$ and let $\gamma = \mathrm{gcd} (j,M)$. Consider,
\begin{align*}
 \mathbb{P} \left ( j Q \, \mathrm{mod} \, M  = k \right ) = \sum_{q\in \mathcal{Q}(M)}  \mathbb{P} \left ( j q \, \mathrm{mod} \, M  = k  | Q = q \right )  \, \mathbb{P}( Q = q ) =   \sum_{q \in \mathcal{Q}(M) } \mathbb{I}_{ j q \, \mathrm{mod} M   = k}(q) \,  \mathbb{P}( Q = q )
\end{align*}
and note that,
\begin{align*}
 \mathbb{P}( Q = q ) = \frac{1}{\# \mathcal{Q}(M) } = \frac{1}{\phi(M)} &\leq \frac{e^\zeta \log(\log(M)) + \frac{3}{\log(\log(M))} }{M} 
\end{align*}
following bounds on the Euler totient function $\phi(\cdot)$ (\cite{ribenboim2012book}), where $\zeta$ is the Euler-Mascheroni constant, and since $Q$ is uniformly distributed in
$\mathcal{Q}(M)$. Therefore,
\begin{align*}
 \mathbb{P} \left ( j Q \, \mathrm{mod} \, M  = k \right ) \leq  \frac{e^\zeta \log(\log(M)) + \frac{3}{\log(\log(M))} }{M}  \,   \sum_q \mathbb{I}_{ j q \, \mathrm{mod} M  =  k}(q) 
\end{align*}
We now show that the quantity $\sum_q \mathbb{I}_{ jq \, \mathrm{mod} M   = k}(q)$ is bounded above and below by,
\begin{equation*}
\gamma - 1 \leq  \sum_q \mathbb{I}_{ jq \, \mathrm{mod} M = k}(q) \leq \gamma
\end{equation*}
To see this, first note that this quantity corresponds to the number of integers $q$ which hash to the integer $k$ through the map $q \rightarrow (j q) \,
\mathrm{mod}\,M$. Now, assume there exists some $q$ such that 
\begin{equation}
\label{1DsFFT:findsupport:modulo}
 j q \, \mathrm{mod} M \equiv k, 
\end{equation}
which implies that
\begin{equation}
\label{1DsFFT:findsupport:diophantine}
j q + i M = k
\end{equation}
for some integer $i \in \mathbb{Z}$. This is a Diophantine equation
which has infinitely many solutions if and only if $\mathrm{gcd}\left (j,M\right ) = \gamma$ divides $k$ (\cite{mordell1969diophantine}). Otherwise, it has no
solution. Assuming it does and $(q_0,i_0)$ is a particular solution, all remaining solutions must take the form 
\begin{equation*}
q = q_0 + u \frac{M}{\gamma}, \; i_0 - u \frac{j}{\gamma}
\end{equation*}
where $u \in \mathbb{Z}$. However, since $0 \leq q < M$ the number of possible solutions must be such that,
\begin{equation*}
\gamma -1 \leq \# \, \left\{ q \in [0,M): q = q_0 + u \frac{M}{\gamma}    \right \} \leq \gamma.
\end{equation*}
which proves the claim. Thus,  
\begin{align*}
\mathbb{P} \left ( jQ \, \mathrm{mod} \, M   = k \right ) &\leq \left ( e^\zeta \log(\log(M)) + \frac{3}{\log(\log(M))} \right )  \,  \frac{ \gamma }{M}
\end{align*}
We can now treat $\mathbb{P} \left ( \left | jQ \, \mathrm{mod} \, M  \right | \leq C \right )$. Before we proceed however, recall that Eq.\eqref{1DsFFT:findsupport:modulo} has a solution if and only if $\gamma | k$. We then write, 
\begin{align*}
\mathbb{P} \left ( \left | jQ \, \mathrm{mod} \, M \right | \leq C   \right ) &= \sum_{ 0 \leq k \leq C } \,\mathbb{I}_{\gamma | k } (k) \, \mathbb{P} \left ( nQ \, \mathrm{mod} \, M  = k \right ) 
\end{align*}
from which it follows that,
\begin{align*}
\mathbb{P} \left ( \left | jQ \, \mathrm{mod} \, M \right | \leq C   \right ) &\leq\left ( e^\zeta \log(\log(M)) + \frac{3}{\log(\log(M))} \right )  \,  \frac{ \gamma }{M} \,  \sum_{ 0 \leq k \leq C } \,\mathbb{I}_{\gamma  | k } (k)  \\
&\leq\left ( e^\zeta \log(\log(M)) + \frac{3}{\log(\log(M))} \right ) \,  \frac{\gamma }{M} \,  \frac{C}{\gamma }  \\
&\leq\left ( e^\zeta \log(\log(M)) + \frac{3}{\log(\log(M))} \right ) \, \frac{C}{M} 
\end{align*}
since the number of integers in $0 \leq k \leq C$ that are divisible by $\gamma$ is bounded above  by $\frac{C}{\gamma }$. Finally, since this holds regardless of our choice of $j$, this proves the desired result.
%
\end{proof}

\begin{lemma}
\label{1DsFFT:findsupport:filter_lemma}
Consider a function $f(x)$ of the form of Eq.\eqref{1DsFFT:orig_func}
and satisfying the constraint Eq.\eqref{1DsFFT:positivity}. Let $0 <
\sigma = \mathcal{O} \left(  \frac{M}{R \sqrt{\log(R)} } \right )  $, $0 < \delta < 1$ and, 
\begin{align*}
\mu &= \min_{j \in \mathcal{S}} |\hat{f}_j| \\
 \Delta &= \frac{\max_{j \in \mathcal{S}} |\hat{f}_j|}{\min_{j \in \mathcal{S}} |\hat{f}_j|}  = \frac{||\hat{f}||_\infty}{ \mu} 
\end{align*}
Finally, let $F_{\mathcal{A}(K;M)} (\cdot)$ and $\Psi_\sigma(\cdot)$ be the operators found in Eq.\eqref{RASR:maineq}. Then, there exists a constant $1 < C < \infty$ such that if,
\begin{equation}
\label{ksize}
K \geq C \, \frac{M \sqrt{\log\left ( \frac{2 \Delta}{\delta} \right ) }  }{\pi \, \sigma} 
\end{equation}
the inequality,
\begin{equation*}
\left [ F_{\mathcal{A}(K;M)} \left (  \Psi_\sigma \left(    f_{k;\,M} \right ) \right )  \right ]_{ n  \frac{M}{K} }   \geq  \delta \, \mu
\end{equation*}
implies that
\begin{equation}
\label{distsize}
\inf_{j \in \mathcal{S}}\left  |  n  \frac{M}{K} - j \right | \leq  \sigma \, \sqrt{  \log\left ( \frac{2 R \Delta}{\delta } \right )}
\end{equation}
and,
\begin{equation*}
\inf_{j \in \mathcal{S}} \left | n \frac{M}{K}  - j \right | \leq \sigma \, \sqrt{  \log\left ( \frac{ 1 }{\delta } \right )},
\end{equation*}
implies that
\begin{equation*}
\left [ F_{\mathcal{A}(K;M)} \left (  \Psi_\sigma \left(    f_{k;\,M}\right ) \right )  \right ]_{n \frac{M}{K} }  \geq  \delta \, \mu
\end{equation*}
for all $n \in \{ 0,1,...,K-1\}$.
\end{lemma}
\begin{proof}
Consider the quantity
\begin{align}
 & \frac{1}{M} \,  \sum_{m \in \mathcal{A}(K;M) } e^{2\pi i \frac{ n\frac{M}{K} \, m}{M}} \, \hat{g}_{ \sigma } \left ( \frac{m}{M} \right ) \, f_{m;M} = \sum_{j \in \mathcal{S} }  \left (  \frac{1}{M} \,  \sum_{m \in \mathcal{A}(K;M) }  e^{2\pi i \frac{m}{K} \left ( n - j  \right )  } \hat{g}_{\sigma} \left ( \frac{m}{M} \right )   \right )\, \hat{f}_j  
\label{filter}
\end{align}
and recall that, 
$$\frac{1}{M}  \hat{g}_{\sigma} \left ( \frac{m}{M} \right ) = \frac{\sqrt{\pi} \sigma}{M}  \sum_{h \in \mathbb{Z} } e^{- \pi^2 \sigma^2 \left ( \frac{m + hM}{M}  \right )^2 } $$
where $ \mathcal{A}(K;M) := \left \{ m  \in \{ 0,1,...,M-1 \} \, : \,m \leq \frac{K}{2} \; \mathrm{or} \;  \left |  m - M \right | <  \frac{K}{2} \right \}$ (Definition \ref{1DsFFT:Fourierset}). From this expression, it is apparent that there exists some constant $1 < C< \infty$ such that by choosing $K \geq C \, \frac{M \sqrt{\log\left ( \frac{2  \Delta}{\delta} \right ) }  }{\pi \, \sigma} $, one has,
\begin{equation*}
 \left |  \sum_{j \in \mathcal{S} }  \left (  \frac{1}{M} \,  \sum_{m \in \mathcal{A}^c (K;M) }  e^{2\pi i \frac{m}{K} \left ( n - j  \right )  } \hat{g}_{\sigma} \left ( \frac{m}{M} \right )   \right )\, \hat{f}_j    \right | \leq   \sum_{m \in \mathcal{A}^c (K;m) } \left (  \frac{\sqrt{\pi} \sigma}{M}  \sum_{h \in \mathbb{Z} } e^{- \pi^2 \sigma^2 \left ( \frac{m + hM}{M}  \right )^2 } \right)   \, || \hat{f} ||_\infty\leq  \frac{\delta \mu }{2} 
\end{equation*}
Indeed, by the integral test,
\begin{align*}
 \sum_{m \in \mathcal{A}^c (K;m) } \left (  \frac{\sqrt{\pi} \sigma}{M}  \sum_{h \in \mathbb{Z} } e^{- \pi^2 \sigma^2 \left ( \frac{m + hM}{M}  \right )^2 } \right)   &\leq A \, \frac{\sqrt{\pi} \sigma}{M} \, \sum_{ m \geq \frac{K}{2} }  e^{- \pi^2 \sigma^2 \left ( \frac{m}{M}  \right )^2 } \\
&\leq A \,  \frac{\sqrt{\pi} \sigma}{M} \,  e^{- \pi^2 \sigma^2 \left ( \frac{K}{2M}  \right )^2 } +   \frac{\sqrt{\pi} \sigma}{M} \, \int_{\frac{K}{2}}^\infty  e^{- x^2 \left ( \frac{\pi \sigma }{M}  \right )^2 } \, \mathrm{d} x \\
&\leq  A \,  \frac{\sqrt{\pi} \sigma}{M} \,  e^{- \pi^2 \sigma^2 \left ( \frac{K}{2M}  \right )^2 }  + \frac{1}{\sqrt{\pi}} \, \mathrm{erfc} \left( \frac{\pi K \sigma}{2M}  \right)  \\
& \leq B \, e^{- \pi^2 \sigma^2 \left ( \frac{K}{2M}  \right )^2 } 
\end{align*}
for some positive constants $A,B$, and where the last inequality follows from estimates on the complementary error function \cite{chang2011chernoff} and the fact that $\frac{\sqrt{\pi} \sigma}{M} = \mathcal{O} \left (\frac{1}{R \sqrt{\log(R)} }  \right ) $ by assumption. Therefore,
\begin{align*}
\left [ F_{\mathcal{A}(K;M)} \left (  \Psi_\sigma \left(    f_{m;\,M} \right ) \right )  \right ]_{n\frac{M}{K}}  &=  \frac{1}{M} \, \sum_{m \in {\mathcal{A}(K;M)} } e^{2\pi i \frac{ n  \,  m}{K}} \, \hat{g}_{ \sigma } \left ( \frac{m}{M} \right ) \, f_{m;M}  \\
&= \frac{1}{M} \,  \sum_{m=0}^{M-1} e^{2\pi i \frac{n \frac{M}{K}  m}{M}} \, \hat{g}_{ \sigma } \left ( \frac{m}{M} \right ) \, f_{m;M}   + \epsilon_n  \\
&= \sum_{j \in \mathcal{S} }  e^{- \frac{\left ( n \frac{M}{K} -j \right ) }{\sigma^2} } \, \hat{f}_j   + \epsilon_n  \\
\end{align*}
where $ \max_n |\epsilon_n|  \leq   \frac{\delta \, \mu }{2}$. Now assume: $|  \left [ F_{\mathcal{A}(K;M)} \left (  \Psi_{ \sigma } \left(    f_{k;\,M}  \right ) \right )  \right ]_{n \frac{M}{K} } |  \geq \delta \, \mu $. Then, the triangle inequality and the previous equation imply that,
\begin{equation}
\label{1DsFFT:findsupport:gaussianfinaleq}
 \sum_{j \in \mathcal{S} }  e^{- \frac{\left (   n  \frac{M}{K}  -j \right ) }{\sigma^2} } \, \hat{f}_j    \geq \left |   \left [ F_{\mathcal{A}(K;M)} \left (  \Psi_{ \sigma } \left(    f_{m;\,M}  \right ) \right )  \right ]_{n \frac{M}{K} }\right |  - \frac{ \delta \, \mu }{2}  \geq  \left ( \delta - \frac{\delta}{2} \right ) \,  \mu= \frac{\delta }{2} \,\mu .
\end{equation}
We claim that this cannot occur unless,
\begin{equation}
\label{sol}
\inf_{j \in \mathcal{S}} \left | n \frac{M}{K}  - j \right | \leq  \sigma \, \sqrt{  \log\left ( \frac{2  R \Delta }{\delta } \right )}.
\end{equation}
We proceed by contradiction. Assume the opposite holds. Then,
\begin{align*}
 \sum_{j \in \mathcal{S} }  e^{- \frac{\left ( n  \frac{M}{K}  -j \right ) }{\sigma^2} } \, \hat{f}_j    \leq ||\hat{f}||_\infty \, \sum_{j \in \mathcal{S}}  e^{- \frac{ \left ( n  \frac{M}{K}  - j \right )^2}{ \sigma^2} }  < ||\hat{f}||_\infty \, \frac{\delta}{2 \Delta  } = \frac{\delta \, \mu }{2} \, 
\end{align*}
by assumption. This is a contradiction. Thus, Eq.\eqref{sol} must indeed hold. This proves the first part of the proposition. For the second part, assume
\begin{equation}
\inf_{j \in \mathcal{S}} \left |  n \frac{ M}{K} - j \right | \leq \sigma \, \sqrt{  \log\left ( \frac{1}{ \delta } \right )}
\end{equation}
holds. Letting $j^*$ be such that $\left | n  \frac{M}{K} -j^*    \right | =  \inf_{j \in \mathcal{S}} \left | n  \frac{M}{K} -j    \right |$, we note that,
\begin{equation*}
 \sum_{j \in \mathcal{S} }  e^{- \frac{\left (  n  \frac{M}{K}  -j \right )^2 }{\sigma^2} } \, \hat{f}_j  \geq  e^{-\frac{\left ( n \frac{M}{K} -j^*   \right )^2}{ \sigma^2} }  \hat{f}_{j^*} \geq \delta \, \mu
\end{equation*}
since $\hat{f}$ and the Gaussian are all positive by assumption. This shows the second part.
\end{proof}


We are now ready to prove the validity of the
Algorithm~\ref{1DsFFT:findsupport}. 
\begin{restatable}{prop}{algsuppcorrectness} {(Correctness of Algorithm \ref{RASR:recoveraliasedsupp})}
\label{1DsFFT:findsupport:findsupport_proof}
Consider a function $f(x)$ of the form of Eq.\eqref{1DsFFT:orig_func}
and satisfying the nonnegativity hypothesis, and let
$\Pi_Q(\cdot)$, $\Psi_\sigma(\cdot)$ and $F_K (\cdot)$ be the
operators found in Eq.\eqref{RASR:maineq}
where $\delta$, $\mu$, $\Delta$ and $K$ are as in Lemma \ref{1DsFFT:findsupport:filter_lemma} and $K$ satisfies the additional constraint $K> \frac{R}{\alpha} \,\sqrt{\frac{\log\left ( \frac{2R\Delta}{\delta} \right ) }{\log\left ( \frac{1}{\delta} \right ) }  } $, and
\begin{align}
\label{sigmadef}
\sigma = \frac{\alpha \, \frac{M}{2R} }{\sqrt{  \log\left ( \frac{2 R \Delta}{\delta} \right )}} 
\end{align}
for some $0<\alpha < 1$. Assume further that the integers $\{ Q^{(l)} \}_{l=1}^L$ are chosen independently and uniformly at random within $\mathcal{Q}(M)$, for some $1\leq L  \in \mathbb{N}$.
Consider 
\begin{equation}
\phi_{ \left [ i \frac{K}{M} \right ] \frac{M}{K} }  (Q^{(l)}) := \left [ F_{\mathcal{A}(K;M)} \left (  \Psi_\sigma \left(    \Pi_{Q^{(l)}}  \left (  f_{  k;M} \right )   \right ) \right )  \right ]_{ \left [ i \frac{K}{M} \right ] \frac{M}{K} } 
\end{equation}
Then,
\begin{equation}
\label{support:mixing}
\mathbb{P} \left ( \cap_{l=1}^L   \left \{  |\phi_{  \left [ i   \frac{K}{M} \right ] \,  \frac{M}{K} } ( Q^{(l)} | \geq \delta \,  \mu \right \} \right ) \leq \alpha^L
\end{equation}
for every $i $ such that $ (i [Q]_M^{-1} ) \, \mathrm{mod} \, M \not \in \mathcal{S}^c$, and
\begin{equation*}
 |\phi_{ \left [ i \,  \frac{K}{M} \right ] \frac{M}{K} } ( Q^{(l)}) | \geq  \delta \, \mu
\end{equation*}
almost surely for all $Q^{(l)}$ and every $i$ such that $ (i [Q]_M^{-1} ) \, \mathrm{mod} \, M  \in \mathcal{S}$.
\end{restatable}

\begin{proof}
From independence, the probability in Eq. \eqref{support:mixing} is equal to, 
\begin{equation*}
\prod_{l=1}^L \, \mathbb{P} \left ( \left  |
\phi_{ \left [ (i Q \, \mathrm{mod} \, M ) \,  \frac{K}{M} \right ] \, \frac{M}{K}}( Q ) \right  | \geq \delta \,  \mu   \right). 
\end{equation*}
So it is sufficient to consider a fixed $l$. Assume first that $i[Q]_M^{-1} \, \mathrm{mod} \, M \not \in \mathcal{S}$. As a consequence of Lemma \ref{1DsFFT:findsupport:filter_lemma} and Lemma \ref{1DsFFT:findsupport:hash} we have the inclusion,
\begin{align*}
\left \{  | \phi_{  \left [ i \,  \frac{K}{M} \right ] \,  \frac{M}{K} } ( Q)  | \geq \delta \,  \mu \right \}  
 &\subset \left \{ \inf_{j \in \mathcal{S}}\left  |  \left [ i  \, \frac{K}{M} \right ] \frac{M}{K} - (j Q ) \, \mathrm{mod} \, M \right | \leq    \sigma \, \sqrt{  \log\left ( \frac{2 R \Delta}{\delta} \right )} \right \} \\ 
  &\subset  \cup_{j \in \mathcal{S}  }  \left \{ \left  | (( ( i[Q]_M^{-1} \, \mathrm{mod} \, M ) -j) Q) \, \mathrm{mod} \, M  \right | \leq    \sigma \, \sqrt{  \log\left ( \frac{2 R  \Delta  }{\delta} \right )} + \frac{M}{2K} \right \}, 
\end{align*}
which implies that the probability for each fixed $l$ is bounded by,
\begin{align*}
\mathbb{P} \left (  | \phi_{ \left [ i   \frac{K}{M} \right ] \, \frac{M}{K} } (Q^{(l)})  | \geq \delta \,   \mu   \right )   &\leq \sum_{j \in \mathcal{S} } \mathbb{P} \left ( \left  | (( i[Q]_M^{-1} \, \mathrm{mod} \, M -j)
    Q^{(l)}) \, \mathrm{mod} \, M   \right | \leq
    \sigma \, \sqrt{  \log\left ( \frac{2 R 
        \Delta}{\delta} \right )} + \frac{M}{2K}\right )  \\
        &\leq
    \mathcal{O} \left ( R \, \, \left ( \frac{   \sigma
      \, \sqrt{  \log\left ( \frac{2 R \Delta  }{\delta} \right )} +
      \frac{M}{2K} }{M} \right )  \right )\\
      &= \mathcal{O}(\alpha), 
\end{align*}
by the union bound, by Lemma~\ref{1DsFFT:findsupport:shuffle} (since $i[Q]_M^{-1} \, \mathrm{mod} \, M \not = j $) and by assumption. Therefore,
\begin{equation*}
\mathbb{P} \left ( \cap_{l=1}^L   \left \{  |\phi_{ \left [ i \frac{K}{M} \right ] \frac{M}{K} } (Q^{(l)})   | \geq \delta  \, \mu \right \} \right )  \leq
\mathcal{O} \left ( \alpha^L   \right ) 
\end{equation*}
as claimed. As for the second part of the proposition, note that if $i[Q]_M^{-1} \, \mathrm{mod} \, M \in \mathcal{S}$ then
\begin{align*}
\inf_{j \in \mathcal{S}} \, \left  |  \left [ i \frac{K}{M} \right ] \frac{M}{K} - (j Q) \, \mathrm{mod} \, M \right | &\leq  \inf_{j \in \mathcal{S}} \, \left  |   \left ( (i [Q^{(l)}]^{-1}_M) \, \mathrm{mod} \, M  - j \right ) Q  \, \mathrm{mod} \, M \right | + \frac{M}{2 K} \\
&= \frac{M}{2 K} \\
& \leq  \sigma \, \sqrt{  \log\left ( \frac{1}{ \delta } \right )}
\end{align*}
by assumption. By Lemma \ref{1DsFFT:findsupport:filter_lemma}, this implies that
\begin{equation*}
 |\phi_{  \left [ i   \frac{K}{M} \right ]  \frac{M}{K} } (Q^{(l)})  | \geq  \delta \, \mu
\end{equation*}
and since this is true regardless of the value of the random variable $Q^{(l)}$, we conclude that it holds almost surely.
\end{proof} 

\begin{remark}  \label{complexityremark} \emph{A careful study of the proof of Lemma \ref{1DsFFT:findsupport:filter_lemma} and Proposition \ref{1DsFFT:findsupport:findsupport_proof} shows that the order $\mathcal{O} \left (R \sqrt{\log(R) } \right )$ size of $K$ arises from the need to bound quantities of the form $ \sum_{j \in \mathcal{S}}  e^{- \frac{ \left ( n\frac{M}{K}  - j \right )^2}{ \sigma^2} } $. In the worst-case scenario (the case treated by Lemma \ref{1DsFFT:findsupport:filter_lemma}), this requires estimates of the form of Eq.\eqref{ksize}, Eq.\eqref{distsize} and Eq.\eqref{sigmadef} which introduce an extra $\sqrt{\log(R)}$ factor in the computational cost (Section \ref{1DsFFT}) relative to the (conjectured) optimal scaling. However, throughout the algorithm the elements of any aliased support $\mathcal{S}_k$ appearing in the sum are always subject to random shuffling first. Lemma \ref{1DsFFT:findsupport:shuffle} states that the shuffling  tends to be more of less uniform. Now, were the elements \emph{i.i.d. uniformly distributed}, it would be easy to show that these quantities are of order $\mathcal{O}(1)$ with high probability, removing the need for the extraneous factor. Unfortunately, our current theoretical apparatus does not allow us to prove the latter. However, following this argument and numerical experiments, we strongly believe that it is possible. In this sense, we believe that through a slight modification of the choice of parameters, our algorithm exhibits an (optimal) $\mathcal{O}\left( R \log(R) \log(N) \right )$ computational complexity with the same guarantees of correctness as the current scheme.}
\end{remark}


\begin{lemma}
\label{1DsFFT:CV:asbound}
Let $\{ P^{(t)} \}_{t=1}^T$ be prime numbers greater than or equal to $R \in \mathbb{N}$, and let $i, j \in \{ 0,1,..., N-1\}$ such that,
\begin{equation*}
i  \, \mathrm{mod} \, P^{(t)} = j  \, \mathrm{mod} \, P^{(t)}, \; \; t = 1, 2, ..., T.
\end{equation*}
If $T > \log_R (N)$, then $i = j$.
\end{lemma}
\begin{proof}
Consider $\{ P^{(t)} \}_{t=1}^T $ as described above and $T > \log_R (N) $, and assume that
\begin{equation*}
i \, \mathrm{mod}\, P^{(t)}  = j \, \mathrm{mod}\, P^{(t)}
\end{equation*}
for $t = 0, 1, ..., T$. This implies in particular that
\begin{equation*}
P^{(t)} \, | \, (j - i) 
\end{equation*}
for $t = 0, 1, ..., T$, and that
\begin{equation*}
\mathrm{lcm}( \{ P^{(t)} \}_{t=1}^T ) \, | \, (j - i).
\end{equation*}
However, since the integers $\{ P^{(t)} \}_{t=1}^T$ are prime (and therefore
coprime), 
\begin{equation*}
\mathrm{lcm}(P^{(t)}) = \prod_{t=1}^T P^{(t)} \geq \left ( \min_t  P^{(t)} \right )^T \geq R^{\log_R (N)} = N.
\end{equation*}
This implies that,
\begin{equation*}
|j - i | \geq N,
\end{equation*}
since $i \not = j$, and this is a contradiction since both belong to $\{ 0, 1, ..., N-1\}$.
\end{proof}

\begin{cor}
\label{1DsFFT:CV:asboundcor}
Let $\{ P^{(t)} \}_{t=1}^T$ are as in Lemma \ref{1DsFFT:CV:asbound} and that $i \not = j, \, k \not = l$, $i,j,k,l \in \{ 0,1,..., N-1\}$ are such that,
\begin{align*}
i  \, \mathrm{mod} \, P^{(t)} &= j \, \mathrm{mod} \, P^{(t)} \\
k  \, \mathrm{mod} \, P^{(t)} &= l \, \mathrm{mod} \, P^{(t)}
\end{align*}
for $t = 1,2,...,T$. Then,
\begin{equation*}
(i-j) = (k-l)
\end{equation*}
\end{cor}
\begin{proof}
The statement is equivalent to,
\begin{equation*}
(i-j)  \, \mathrm{mod} \, P^{(t)} = 0 =  (k-l) \, \mathrm{mod} \, P^{(t)}
\end{equation*}
for $t = 1,2,... T$. By Lemma \ref{1DsFFT:CV:asbound}, this implies that $(j-i) = (k-l)$.
\end{proof}

\begin{prop}
\label{1DsFFT:CV:Neumann}
Let $0<R< N \in \mathbb{N}$. Further let $\{ P^{(t)} \}$ be random integers uniformly distributed within the set $\mathcal{P}$ containing the $4 R \log_R(N)$ smallest prime numbers strictly larger than $R$, and let $F$ and $B$ be defined as in Eq.\eqref{1DsFFT:CV:system} with these parameters. If $T \geq 4$, then,
\begin{equation*}
\mathbb{P} \left ( \left | \left |  ( I -  \frac{1}{T} (F B)^* (F B) ) x \right | \right |_2  > \frac{1}{2} \right )   \leq \frac{1}{2}
\end{equation*}
\end{prop}
\begin{proof}
First, note that,
\begin{equation*}
 (FB)^* (FB) =  B^* F^* F B =  B^* B
\end{equation*}
since $F$ is a block-diagonal Fourier matrix, and $I - \frac{1}{T} B^* B$ has entries
\begin{equation}
\label{1DsFFT:CV:Neumann:EqBdef}
 \left [ I - \frac{1}{T} B^{(t)^*} B^{(t)} \right ]_{ij} = \delta_{i,j} - \frac{1}{T} \sum_s B^{(t)}_{si}  B^{(t)}_{sj}=
\left\{
	\begin{array}{ll}
		\frac{1}{T}  & \mbox{if } \, i  \, \mathrm{mod}\, P^{(t)} = j  \, \mathrm{mod}\, P^{(t)}   \\
		0 & \mbox{o.w. }
	\end{array}
\right.
\end{equation}
Therefore, for any vector $x$ such that $||x||_2 = 1$,
\begin{align*}
\mathbb{P} \left ( \left | \left |  ( I - \frac{1}{T} (F B)^* (F B) ) x \right | \right |_2  > \frac{1}{2} \right )  &\leq 4 \, \mathbb{E} \left [ \left (  \sum_{i\not = j}  \sum_t  \bar{x}_i  \, [B^{(t)^*} B^{(t)}]_{ij} \, x_j \right )^2 \right ] \\
&= 4 \, \sum_{i\not = j}  \sum_{k\not = l}  \bar{x}_i\, x_j \,  x_k \,\bar{x}_l   \, \sum_{s,t}   \, \mathbb{E} \left  [ [B^{(t)^*} B^{(t)}]_{ij}  \, [B^{(s)^*} B^{(s)}]_{kl}\right ]
\end{align*}
by Chebyshev inequality.
Furthermore, thanks to Eq.\eqref{1DsFFT:CV:Neumann:EqBdef} and independence, the expectation can be written as,
\begin{align}
\label{1DsFFT:CV:Neumann:EqCond}
\mathbb{E} \left  [ [B^{(t)^*} B^{(t)}]_{ij}  \, [B^{(s)^*} B^{(s)}]_{kl}\right ] &=
\left\{
	\begin{array}{ll}
		\mathbb{P}\left ( \left \{ (i-j)   \,\mathrm{mod} \, P^{(t)} = 0  \right \} \right ) \, \mathbb{P}\left ( \left \{ (k-l)    \,\mathrm{mod} \, P^{(t)} = 0  \right \} \right )   & \mbox{if } s \not = t  \\
		\mathbb{P}\left ( \left \{ (i-j)    \,\mathrm{mod} \, P^{(t)} = 0  \right \}   \cap  \left \{  (k-l)   \,\mathrm{mod} \, P^{(t)}  = 0  \right \} \right )  & \mbox{if } s = t
	\end{array}
\right.
\end{align}
Now, let $\tau(i,j)$ be defined as
\begin{equation*}
\tau (i,j) := \left \{ P^{(t)} \in \mathcal{P} : i   \, \mathrm{mod} \, P^{(t)} = j   \, \mathrm{mod} \, P^{(t)}   \right \}.
\end{equation*}
The case $s \not = t$ is treated as follows,
\begin{align*}
 &    \mathbb{P}\left (  \left \{  (i-j)   \,\mathrm{mod} \, P^{(t)} = 0  \right \}  \right ) \, \mathbb{P}\left ( \left \{  (k-l)   \,\mathrm{mod} \, P^{(s)} = 0  \right \}  \right ) \\
  &= \left ( \sum_{p_1 \in \tau (i,j) } \mathbb{P}\left ( \left . \left \{  (i-j)    \,\mathrm{mod} \, P^{(t)} = 0  \right \}  \right | P^{(t)} = p_1 \right ) \mathbb{P} \left (  P^{(t)}  = p_1 \right ) \right )\, \cdot  \\
  & \left (  \sum_{p_2 \in \tau (k,l) }  \mathbb{P}\left ( \left . \left \{  (k-l)   \,\mathrm{mod} \, P^{(s)} = 0  \right \}  \right | P^{(s)} = p_2 \right ) \mathbb{P} \left (  P^{(s)} = p_2 \right )  \right ) \\
  &\leq  \left ( \frac{\#\tau (i,j)}{4R \log_R(N)} \, \frac{\#\tau (k,l)}{4R \log_R(N)}  \right )\\
  & \leq  \frac{1 }{ 16 R^2  }
\end{align*}
since $P^{(t)}$ is uniformly distributed within a set of cardinality $4R \log_R(N) $, and because,
\begin{equation*}
 \sum_{p_1 \in \tau (i,j)} \mathbb{P}\left (  \left . \left \{  (i-j)   \,\mathrm{mod} \, P^{(t)} = 0  \right \}  \right | P^{(s)}  = p_1 \right )   = \# \tau (i,j) = \log_R(N) 
\end{equation*}
by Lemma \ref{1DsFFT:CV:asbound}. This leaves us the case $s=t$. To this purpose, we further split this
case into two subcases: that when $i-j = k-l$ and that when $i-j \not
= k-l $. When $i-j = k-l$ we obtain,
\begin{align*}
& \sum_{s,t = 1}^T   \mathbb{P}\left (  \{s=t\} \cap \{i-j  = k-l \} \cap \left \{ (i-j)  \,\mathrm{mod} \, P^{(t)} = 0  \right \}   \cap  \left \{  (k-l)  \,\mathrm{mod} \, P^{(s)}  = 0  \right \} \right ) \\
 &=  \sum_{ p \in \tau(k,l)} \mathbb{P}\left (   \left . \left \{  (k-l)  \,\mathrm{mod} \, P^{(t)}  = 0  \right \} \right | P^{(t)} = p \right ) \, \mathbb{P} \left (P^{(t)} = p \right )  \\
& \leq   \frac{1   }{4R}
\end{align*}
since $k \not = l$, following an argument similar to the previous one. This leaves the case $s=t$, $i-j \not = k-l $. However, thanks to Corollary \ref{1DsFFT:CV:asboundcor} it follows that the set,
\begin{equation*}
  \{s=t\} \cap \left \{ i \not =j  \right \} \cap \left \{ k \not = l  \right \} \cap \left \{ i-j \not = k-l  \right \} \cap \left \{ (i-j)  \,\mathrm{mod} \, P^{(t)} = 0  \right \}   \cap  \left \{  (k-l)   \,\mathrm{mod} \, P^{(s)}  = 0  \right \}
\end{equation*}
must be empty. Putting everything together we find that,
\begin{align*}
\mathbb{P} \left ( \left | \left |  ( I - \frac{1}{T} (F B)^* (F B) ) x \right | \right |_2  > \frac{1}{2} \right ) &\leq 4  \sum_{i\not = j}  \sum_{k\not = l}    \bar{x}_i\, x_j \,  x_k \,\bar{x}_l   \frac{1}{T^2}  \left [   \sum_{s,t = 1}^T  \, \left (  \mathbb{E} \left  [  \mathbb{I}_{s \not = t } (s,t) \, [B^{(t)^*} B^{(t)}]_{ij}  \, [B^{(s)^*} B^{(s)}]_{kl}\right ] + \right . \right . \\
& \left. \left .    \mathbb{E} \left  [   \mathbb{I}_{s = t } (s,t) \, \mathbb{I}_{i-j = k-l} (i,j,k,l)  \, [B^{(t)^*} B^{(t)}]_{ij}  \, [B^{(s)^*} B^{(s)}]_{kl}\right ]\right) \right ]  \\
&\leq   4     \left ( \frac{1  }{ 16 R^2  }  \right ) \left (  \sum_{k\not = l} \bar{x}_l \, x_k \right )^2 +  \frac{ 4  }{T}  \,  \left (  \frac{1  }{4R} \right )  \,  \left (  \sum_{k\not = l} \bar{x}_l \, x_k \right ) \,  \left (  \sum_{j} \bar{x}_{j+k-l} \, x_j \right )  \\
\end{align*}
We further note that $\sum_{i\not = j} \bar{x}_l \, x_k $ is a bilinear form bounded by the norm of an $R \times R$ matrix with all entries equal to $1$ except the diagonal which is all zeros. It is easy to work out this norm which is equal to $R-1$ so that,
\begin{equation*}
\frac{1}{R}\,  \sum_{k\not = l} \bar{x}_l \, x_k < 1
\end{equation*}
Finally, by Cauchy-Schwartz inequality,
\begin{equation*}
 \left | \sum_{j} \bar{x}_{j+k-l} \, x_j  \right | \leq \sqrt{ \sum_{j} |x_{j+k-l}|^2 } \, \sqrt{ \sum_{j} |x_{j}|^2 } = || x||_2^2 = 1.
\end{equation*}
 Thus,
\begin{equation*}
\mathbb{P} \left ( \left | \left |  ( I - \frac{1}{T}  (F B)^* (F B) ) x \right | \right |_2  > \frac{1}{2} \right ) <   \frac{1}{4} + \frac{1}{T}   \leq \frac{1}{2}
\end{equation*}
as claimed.
\end{proof}

\begin{cor}
\label{1DsFFT:CV:NeumannCor}
Under the hypotheses of Proposition \ref{1DsFFT:CV:Neumann}, the solution to the linear system
\begin{equation*}
FB \, \hat{f} = f_0
\end{equation*}
takes the form,
\begin{equation*}
\hat{f} = \sum_{n=0}^{\infty} \left  [ I - \frac{1}{T} B^* B \right  ]^n \,  \left (  \frac{1}{\sqrt{T}} (F B)^*  f_0  \right )
\end{equation*}
with probability at least $\frac{1}{2}$.
\end{cor}
\begin{proof}
By Proposition \ref{1DsFFT:CV:Neumann}, $|| I - \frac{1}{T} (F B)^* (F B)||_2 < \frac{1}{2}$ with probability at least $\frac{1}{2}$. When this is the case we write,
\begin{equation*}
FB \, \hat{f} =  f_0 \Leftrightarrow  \frac{1}{T} B^* B \, \hat{f} =  \frac{1}{\sqrt{T}} B^* F^*  f_0   \Leftrightarrow \left [ I - \left ( I - \frac{1}{T} B^* B \right ) \right ] \, \hat{f} = \frac{1}{\sqrt{T}}  (FM)^*  b  =  \hat{f}_0
\end{equation*}
In this case, it is easy to verify that the Neumann series,
\begin{equation*}
\hat{f} =  \sum_{n=0}^{\infty} \left [ I - \frac{1}{T} B^* B \right ]^n \,  \left (  \frac{1}{\sqrt{T}} (F B)^* b  \right )
\end{equation*}
satisfies this last equation, and that the sum converges exponentially fast.
\end{proof}

\ValRecProp*
\begin{proof}
By Proposition \ref{1DsFFT:CV:Neumann}, $\frac{1}{T} (FB)(FB)^* = I - \mathcal{P} $ where $\left | \left | I -\frac{1}{T} (FB)(FB)^*  \right | \right |_2 = || \mathcal{P} ||_2 < \frac{1}{2}$ with probability larger than $ \frac{1}{2}$. Thus, if we consider $\mathcal{O}(\log_\frac{1}{2}(p))$ independent realizations of $FB$, the probability that at least one of them is such is greater than or equal to $(1-p)$. When this occur, Corollary \ref{1DsFFT:CV:NeumannCor} states that the solution is given  by the Neumann series. Furthermore,
\begin{align*}
\left | \left | \hat{f} -  \sum_{n=0}^{\left \lceil \log_{\frac{1}{2} }(\eta) \right \rceil }  \mathcal{P}^n \, f^\dagger  \right | \right |_2  &= \left | \left |  \sum_{\left \lceil \log_{\frac{1}{2} } (\eta) \right \rceil }^\infty  \mathcal{P}^n \, f^\dagger  \right |  \right |_2  \\
&\leq  \sum_{\left \lceil \log_{\frac{1}{2} } (\eta ) \right \rceil }^\infty  || \mathcal{P} ||_2^n \, || f^\dagger ||_2  \\
&\leq \mathcal{O}(\eta)
\end{align*}
by the geometric series and the bound $|| \mathcal{P} ||_2 \leq \frac{1}{2}$.\\
\end{proof}

\subsection{Proofs of Section \ref{stability}}

\stabilitysuppthm*
\begin{proof}
First, not that since $ \Pi_Q (\cdot)$ is an isomorphic permutation operator (for all $Q \in \mathcal{Q}(M_k)$) one has
\begin{equation*}
|| \Pi_Q ||_\infty = 1.
\end{equation*}
Similarly, since the filtering operator $\Psi_\sigma (\cdot)$ is
diagonal with nonzero entries $\hat{g}_\sigma (n)$, then 
\begin{align*}
|| \Psi_\sigma ||_\infty = \sup_{m \in \{ 0,1,...,M_k-1\} }  \left | \hat{g}_\sigma \left ( \frac{m}{M_k} \right ) \right  | &\leq \sqrt{\pi} \, \sigma = \frac{ \sqrt{\pi} \frac{\alpha M_k}{2R} }{\sqrt{\log\left ( \frac{2 R \Delta}{\delta}  \right ) } }.
\end{align*}
Finally, we get from the triangle inequality that,
\begin{align*}
\leq  \frac{\# \mathcal{A}(K;M_k) }{M_k} \,|| \Psi_\sigma ||_\infty \, || \Pi_Q ||_\infty \,\left| \left| \nu \right | \right |_\infty
& \leq  \frac{ \sqrt{\pi}  \frac{\alpha K}{ R} }{\sqrt{\log\left ( \frac{2 R \Delta}{\delta}  \right ) } } \, \left| \left| \nu \right | \right |_\infty \\
&\leq \frac{  \sqrt{\pi} \frac{\alpha K}{ R} }{\sqrt{\log\left ( \frac{2 R \Delta}{\delta}  \right ) } } \, || \hat{\nu}||_1
\end{align*}
by the Hausdorff-Young inequality \cite{beckner1975inequalities}. Finally, we note that: $|| \hat{\nu}||_1 \leq \sqrt{N} \, ||\hat{\nu}||_2 < \eta$ by assumption, and recall that $K = \mathcal{O}(R \sqrt{\log(R)}) $. This leads to the desired result.


%

%
\end{proof}

\subsection{Proof of Section \ref{MD}}
 
\mdsfft*
\begin{proof}
First, note that
\begin{equation*}
\int_{[0,1]^d} e^{-2\pi i \, j\cdot x} \, f(x) \, \mathrm{d} x  = \hat{f}_j.
\end{equation*}
Then, substitute the samples in the quadrature to obtain 
\begin{align*}
 \frac{1}{N} \sum_{n=0 }^{N-1}e^{-2\pi i \, j \cdot  x_n} \, f(x_n) &=  \frac{1}{N} \sum_{n =0}^{N-1} e^{-2\pi i \, j \cdot \frac{n g \, \mathrm{mod} N}{N} } \, \left ( \sum_{k \in [0,M)^d  \cap \mathbb{Z}^d }  \hat{f}_k \, e^{2 \pi i \,   k  \cdot \frac{n g \, \mathrm{mod} N}{N}  } \right ) \\
 &= \sum_{k \in [0,M)^d  \cap \mathbb{Z}^d  }  \hat{f}_k \left (   \frac{1}{N} \sum_{n=0
   }^{N-1} e^{-2\pi i \,   \frac{n ( (k-j)\cdot g)}{N} } \right) 
\end{align*}
since $e^{2 \pi i \,   (k-j)  \cdot \frac{n g \, \mathrm{mod} N}{N}  } =
e^{2 \pi i \,   (k-j)  \cdot \frac{n g }{N}  }$. Note however that 
\begin{equation*}
  \frac{1}{N} \sum_{n=0 }^{N-1} e^{-2\pi i \,   \frac{n ((k-j)\cdot g)}{N}
  } = D_N \left (  (k-j) \cdot g   \right ), 
\end{equation*}
which is the Dirichlet kernel and is equal to $0$ unless $(k-j) \cdot g = 0 \,  \mathrm{mod} \, N$, in which case it is equal to $1$. Thus,
\begin{equation*}
 \frac{1}{N} \sum_{n=0 }^{N-1} f(x_n)   =  \hat{f}_j +  \sum_{
   \substack{ k \in [0,M)^d \cap \mathbb{Z}^d  \\ (k-j) \cdot g  \,
     \mathrm{mod} \, N \equiv 0 \\ (k-j) \cdot g  \not = 0  } }
 \hat{f}_k. 
\end{equation*}
Thus, in order to show that the quadrature is exact, it suffices
to show that the remaining sum on the right-hand side of the previous
equation is trivial. To see this, note that $(k-j) \in [-M,M)^d \cap \mathbb{Z}^d$ and consider
\begin{equation*}
| (k-j) \cdot g| = \left | (k_1 - j_1) + (k_2 - j_2) M + ... + (k_{d} - j_d) M^{d-1} \right |
\leq M \sum_{l=0}^{d-1} M^l = M \, \frac{1 - M^d}{1-M} < M^d = N, 
\end{equation*}
where the inequality is strict for any finite $M \in
\mathbb{N}$ strictly larger than $1$. This implies that there cannot be any $(k-j)$ other than $0$ in the domain of interest such that $(k-j) \cdot g  \, \mathrm{mod} \, N
\equiv 0$. The sum is therefore empty and the result follows. 
\end{proof}

\section{Generalization to general complex sparse vectors}
\label{complexvec}
This appendix provides a terse description of the additional steps necessary to transform the sMFFT for \emph{real positive vectors}  into a reliable algorithm for \emph{general complex vectors}. To achieve this task, two major hurdles, both associated with the support-recovery portion of the scheme, must be overcome; the first one is associated with the initial aliasing of the signal described in Section \ref{1DsFFT}. As shown Eq.\eqref{1DsFFT:aliasing}, at each step aliasing implies Fourier coefficients of the form,
\begin{equation*}
\hat{f}^{(k)}_l = \sum_{j : j \, \mathrm{mod} \, M_k = l } \, \hat{f}_j , \; l = 0,1, ..., M_k .
\end{equation*}
When the original nonzero coefficients are all strictly positive, this expression is positive \emph{if and only if} the lattice $l + i M_k, \; i = 0, 1, ... \frac{N}{M_k}-1$
contains one of the original nonzero coefficients. When the nonzero coefficients are complex however, this is no longer true. The second potential issue pertains to the resulting filtering step found in Algorithm \ref{RASR:recoveraliasedsupp}. As described by Eq.\eqref{support:smoothFT}, the result takes the form,
\begin{equation*}
\left [ \Psi_\sigma \left ( \Pi_Q \left ( f_{n;\,M_k} \right ) \right ) \right ] (\xi)  = \mathcal{F}^* \left [ \sum_{j\in\mathcal{S}_k} \hat{f}^{(k)}_j \,  e^{- \frac{ \left | x  -  \left ( j [Q]^{-1}_{M_k} \, \mathrm{mod} \, M_k \right ) \right  |^2}{\sigma^2} } \right ] (\xi). 
\end{equation*}
which corresponds to the Fourier transform of the aliased signal convoluted with a Gaussian. Once again, the crucial statistical test used in Algorithm \ref{RASR:recoveraliasedsupp} relies on 
this quantity being positive if and only if a point lies in the vicinity of an element of the (shuffled and aliased) support $\mathcal{S}_k$. Such statement does not hold true if we allow the coefficients to be general complex numbers (as some elements might \emph{cancels out}).

The conclusion of these observations is that as a consequence of the lack of positivity, it is possible that elements belonging to $\mathcal{M}_k \cap \mathcal{S}_k$ might be wrongfully eliminated in Algorithm \ref{RASR:recoveraliasedsupp}, i.e., the \emph{false negative identification rate is nontrivial}. To alleviate these issues, we propose a slight modification to the scheme; we allow for the possibility of the output of Algorithm \ref{RASR:recoveraliasedsupp} be missing elements of $\mathcal{S}_k$ by launching multiple independent runs of the $\mathrm{FIND\_\,SUPPORT}(\cdot)$ routine in Algorithm \ref{1DsFFT:algorithm}, and taking the \emph{union} of the outputs. In this sense, although it is possible to miss an element with a single run, we expect that the probability of a miss over multiple independent run is very small. In addition, this modification \emph{does not have any effect on the fundamental computational complexity}; indeed, close examination shows that these additional steps \emph{only increase the algorithmic constant by some small quantity independent of $N$ and/or $R$}.

So far, this modification remains a heuristic (with some preliminary/unpublished theoretical backing). Note however that we have implemented it and can attest to excellent numerical results in line with our expectation based on the previous discussion, and very similar to those obtained in the real-positive case.

\newpage
\bibliographystyle{plain}
\bibliography{biblio}

\end{document}